\documentclass{article}
\usepackage{amsmath}
\usepackage{amsthm}
\usepackage{amssymb}
\usepackage{graphicx}
\usepackage{xcolor}

\newtheorem{theorem}{Theorem}[section]
\newtheorem{proposition}[theorem]{Proposition}
\newtheorem{lemma}[theorem]{Lemma}
\newtheorem{corollary}[theorem]{Corollary}
\newtheorem{definition}[theorem]{Definition}
\newtheorem{example}[theorem]{Example}
\newtheorem{remark}[theorem]{Remark}

\DeclareMathOperator{\dist}{dist}

\begin{document}

\title{On the space of compact diamonds of Lorentzian length spaces}

\author{Waldemar Barrera, Luis Montes de Oca, Didier A. Solis}

\date{\today}

\maketitle

\abstract{
In this work we introduce the taxicab and uniform products for Lorentzian pre-length spaces. We further use these concepts to endow the space $\mathcal{D}(\mathbb{R}\times_T X)$ of causal diamonds with a Lorentzian length space structure, closely relating its causal properties with its geometry as a metric space furnished with its associated Hausdorff distance. Among the general results, we show that this space is geodesic and globally hyperbolic for complete length space $(X,d)$.}

\medskip

\textbf{Keywords:} \textit{Lorentzian length spaces, Causality, Hyperspaces}

\medskip

\textbf{MSC Classification} \textit{53C23, 53C22, 51F30, 54F16, 83C99}

\section{Introduction}\label{sec1}

In the past few years there has been a remarkable interest in exploring synthetic alternatives to well established  notions arising in Lorentzian geometry.  The seminal work of Kunzinger and S\"amann on Lorentzian length spaces \cite{KS} has paved the way to introduce such methods in the realm of Relativity, in which other promising approaches are currently been developed, as is the case of Lorentzian metric spaces \cite{minsur} or almost Lorentzian length spaces \cite{Olaf2} just to mention a few. These developments are suitable candidates to provide useful tools for the  analysis of physical scenarios were the standard regular (at least $C^2$) framework is no longer assumed, either from the theoretical onset ---for instance, when considering $C^{1,1}$ \cite{Graf1,Graf2,Kun1,Kun2}, $C^1$ \cite{Graf3,Kun3}  or even $C^0$ \cite{galloling,grantc0,grantin,hevelin02,sbierski} metrics--- or in the interest of further our understanding of the most recent observation of black hole merging and gravitational waves \cite{ebh,ligo}.

A cornerstone in Kunzinger and S\"amann development consists in the axiomatization of the fundamental causal properties of spacetime through the notion of a \textit{Lorentzian pre-length space}, which can be regarded as a refinement of the concept of causal space first introduced by Kroenheimer and Penrose in the early years of Mathematical Relativity \cite{kpcs}. As an example of the value of building a causal theory on axiomatic grounds we cite the causet approach to Quantum Gravity (see \cite{surya} and references therein), in which no spacetime metric is considered at all. The existence of a time separation function, as proposed in \cite{KS}, opens the scope of applications to settings where no regular spacetime metric is given a priori, as is the case of Lorentzian manifolds with timelike boundary \cite{ACS}, causal completions of globally hyperbolic spacetimes \cite{ABS} or even some class of contact structures \cite{hendike}.

Moreover, by considering a suitable class of causal curves, a synthetic notion of curvature bounds is defined and its relation to its geodesic structure analyzed in a way akin to the well established theory of Alexandrov and CAT(k) metric length spaces \cite{burago}. In this regard, several results inspired in this theory are currently been pursued in the Lorentzian pre-length space scenario. For instance, Toponogov type triangle comparison \cite{solismontes}, or the existence of hyperbolic angles and exponential maps \cite{beran01}. Applications of these results include gluing and amalgamation \cite{beran02}, as well as a splitting theorem that generalizes the landmark result for Lorentzian manifolds with non-positive timelike sectional curvature \cite{splitting}.

Indeed, curvature bounds are essential in some of the most remarkable results pertaining to metric length spaces, such as the celebrated pre-compactness theorems of Gromov \cite{GromovBook,Petersen}. Recall that these results can be interpreted in terms of stability of curvature bounds with respect to Gromov-Hausdorff convergence.  Thus, it is natural to explore the stability of different aspects of the theory of Lorentzian length spaces under suitable notions of convergence. For instance,  in \cite{kunzinger02} Kunzinger and Steinbauer address the behaviour of curvature bounds under Gromov Hausdorff convergence of a class of warped product spacetimes via the null distance first introduced by Sormani and Vega \cite{Sormani}. Another key aspect of the theory closely tied to the notion of Gromov distance was studied in \cite{MCann}, where McCann and S\"amann establish notions of Hausdorff measure and dimension in the context of Lorentzian length spaces. Let us emphasize that in this work causal diamonds play a fundamental role just as (compact) metric balls do in the context of metric length spaces.

In this work we aim at providing a first glance on the structure of the space of causal diamonds of a globally hyperbolic Lorentzian pre-length space $(X,d,\ll ,\le ,\tau)$, furnished with the Hausdorff distance $d_H$ induced by $d$. Recall that the metric spaces thus obtained from families of compact subsets of a metric space --or isometry classes thereof-- are commonly referred to as \emph{hyperspaces} and their study is cornerstone in the theory of continua \cite{Nadler1, Nadler2} as well as infinite dimensional topological manifolds and function spaces \cite{Chapman,VanMill}. From the topological point of view, the celebrated theorem of Curtis establishes that the hyperspace $\mathcal{H}(\mathbb{R}^n)$ of compact subsets of  $\mathbb{R}^n$ is homeomorphic to the punctured Hilbert cube \cite{Curtis}, and a similar characterization holds for the hyperspace $\mathbb{GH}(\mathbb{R}^n)$ of isometry classes under the Gromov Hausdorff distance \cite{sergey1,sergey2}. In spite of these accurate topological descriptions, the geodesic structure of hyperspaces remains unexplored for the most part. For instance, only recently  the intrinsic nature of $\mathbb{GH}(\mathbb{R}^n)$ has been established \cite{Memoli,Ivanov}. Moreover, by the well known Gromov's Embedding Lemma, a family of compact metric spaces can be embedded in a $\ell^\infty$ space endowed with its Hausdorff distance \cite{Gromov}. Thus, exploring the geodesic structure of the space of causal diamonds with respect to the Hausdorff distance is an important step in the program of understanding the geometry of its corresponding Gromov Hasudorff hyperspace. By focusing on the class of Lorentzian taxicab spaces we provide a precise description of the space of causal diamonds of Minkowski space and establish that it has the structure of a globally hyperbolic Lorentzian length space, and thus it satisfies de Avez-Seifert property. Moreover, we give explicit parametrizations for maximal curves joining any pair of causally related diamonds.

This paper is organized as follows. In Section \ref{sec:prelim} we state the basic definitions and establish the notation to be used throughout this work. In Section \ref{sec:products} we describe the Lorentzian taxicab product of a Lorentzian pre-length space and a metric space; as well as the uniform Lorentzian product of two Lorentzian pre-length spaces, and study in more detail the case of the Lorentzian taxicab product $\mathbb{R}^1_1\times_T X$. In Section \ref{sec:diamonds} we prove that the hyperspace of compact subsets of a Lorentzian pre-length space can be endowed with such a structure. Finally, as application of this fundamental result, in Section \ref{sec:apps} we study the space of causal diamonds $\mathcal{D}(\mathbb{R}^1_1\times_T X)$ and prove that it admits a structure of globally hyperbolic Lorentzian Length space, give explicit parametrization for a family of maximal timelike curves and provide an explicit realization as a subspace of a Lorentzian uniform product.

\section{Preliminaries}\label{sec:prelim}

In this section we set up the basic notions on Lorentzian length spaces and the Hausdorff distance.

\subsection{Lorentzian length spaces}

There are two fundamental aspects in mathematical relativity that are rooted in Lorentzian geometry. In one hand we have the causal structure, and in the other, the locally Minkowskian nature of spacetime; the former coming from the causal character of vectors and the latter from properties of the exponential map. Any attempt to axiomatize  the fundamental properties of spacetime must take into account these issues in a way that does not depend on the existence of a smooth metric tensor. In their seminal work  Kunzinger and S\"amann develop an original notion of Lorentzian length of causal curves by encoding causality in the definition of Lorentzian pre-length space and by describing the local structure via the so called localizing neighborhoods (see \cite{KS} for a detailed account).

\begin{definition}\label{defi:lpls}
A Lorentzian pre-length space is quintuple $(X,d,\ll ,\le ,\tau)$ where $(X,d)$ is a metric space and
\begin{enumerate}
\item $\ll$ and $\le$ are two relations in $X$ that satisfy
\begin{itemize}
\item $\le$ is a pre-order.
\item $\ll$ is a transitive relation containing $\le$.\footnote{In \cite{kpcs} Kronheimer and Penrose define a \emph{causal space} as a triple $(X,\ll ,\le )$ that satisfies properties property (1) plus the causality axiom: $x\le y$ and $y\le x$ $\Leftarrow$ $x=y$.}
\end{itemize}
\item $\tau : (X,d)\to [0,\infty]$ is a lower semi-continuous function for which the reverse triangle inequality holds:
\[
\tau(x,z)\ge \tau (x,y)+\tau (y,z),\quad \forall x\le y\le z .
\]
\end{enumerate}
\end{definition}
In this setting, $\ll$ and $\le$ are the \textit{chronological} and \textit{causal} relations, respectively; whereas $\tau$ is called a \textit{time separation function.} The relations $\ll$, $\le$ enable us to define the chronological and causal sets in the standard way:
\begin{eqnarray*}
I^+(x)=\{y\in X\mid x\ll y \},&\quad& J^+(x)=\{y\in X\mid x\le y\},\\
I^-(x)=\{y\in X\mid y\ll x\},&\quad& J^+(x)=\{y\in X\mid y\le x\}.\\
\end{eqnarray*}

\begin{example}
The set of real numbers $X=\mathbb{R}$ is a Lorentzian pre-length space when we define the chronological and casual relations by $<$ and $\le$, respectively. Here $d(x,y)=\vert y-x\vert$ and
\[
\tau (x,y)=
\left\{ 
\begin{array}{cl}
y-x & \textrm{if } x<y,\cr
0 & \textrm{otherwise.}
\end{array}
\right.
\]
In what follows, we will denote this Lorentzian pre-length space by  $\mathbb{R}^1_1$.
\end{example}

\begin{example}
Any smooth spacetime $(M,g)$ is a Lorentzian pre-length space when we consider any metric $d$ inducing the manifold topology. The relations $\ll$ and $\le$, as well as the time separation function $\tau$ are the standard ones constructed from the Lorentzian metric $g$.   
\end{example}

Lorentzian pre-length spaces have just enough structure to guarantee the most basic results in causality theory. As straightforward consequence of the definition we have the following lemma.
\begin{lemma}
Let $(X,d,\ll ,\le ,\tau)$ be a Lorentzian pre-length space, then
\begin{enumerate}
\item The chronological sets $I^+(x)$, $I^-(x)$ are open in $(X,d)$.
\item If $x\le y\ll z$ or $x\ll y\le z$ then $x\ll z$. (Push-up property).
\end{enumerate}
\end{lemma}
The relations $\ll$, $\le$ can be used to define the class of causal curves in a Lorentzian pre-length space. We say that a Lipschitz continuous curve $\alpha : I\to (X,d)$ is \textit{future timelike (causal)} if for any $t<s$ we have $\alpha (t)\ll \alpha (s)$ ($\alpha (t)\le \alpha (s)$). Past timelike and causal curves are defined analogously. If a curve is causal and no pair of its points are chronologically related we call it \textit{null}\footnote{Observe that if a smooth spacetime $(M,g)$ does not have closed causal curves then this notion of causal character of curves coincides with the usual one.}. Notice that since the relations $\ll$, $\le$ are not defined in terms of curves, it might happen that causally related points can not be joined by any causal curve. A Lorentzian pre-length space $(X,d,\ll ,\le ,\tau )$ for which any chronologically (causally) related pair of points there exists a timelike (causal) curve connecting them is called \textit{causally path connected}.

Normal neighborhoods are essential for describing the local geometry of a smooth spacetime. The analog for Lorentzian length spaces is the concept of \textit{localizing neighborhoods}.

\begin{definition}
Let $(X,d\ll ,\le ,\tau )$ be a causally path connected Lorentzian pre-length space. A neighborhood $\Omega_x$ around  $x\in X$ is localizing if  there exists a continuous functions $\omega_x:\Omega_x\times\Omega_x\to [0,\infty)$ for which 
\begin{enumerate}
\item $(\Omega_x,d\vert_{\Omega_x\times\Omega_x},\ll_{\Omega_x},\leq_{\Omega_x}, \omega_x)$ is a Lorentzian pre-length space, where
$\le_{\Omega_x}$, $\ll_{\Omega_x}$ denote the relations induced by $\ll$, $\le$ in $\Omega_x$\footnote{That is, $p\ll_{\Omega_x}q$ ($p\le_{\Omega_x}q$) if and only if there exists a future timelike (causal) curve in $\Omega_x$ from $p$ to $q$}.
\item $I^\pm (y)\cap \Omega_x\neq\emptyset$, for all $y\in\Omega_x$.
\item All causal curves contained in $\Omega_x$ have uniformly bounded $d$-length.
\item For all $p\neq q\in \Omega_x$ with $p\le q$ there exists a future causal curve $\gamma_{pq}$ contained in $\Omega_x$ such that $L_\tau(\gamma_{pq})=\omega_x(p,q)$ and whose $\tau$-length is maximal among all future causal curves from $p$ to $q$ lying in $\Omega_x$.
\end{enumerate}
If in addition, the following holds: if $p\ll q$ then $\gamma_{pq}$ is timelike and $L(\gamma_{pq})>L(\gamma )$ for each causal curve containing a null segment, then we say the neighborhood is \emph{regular}. 
\end{definition}

Roughly speaking, a Lorentzian length space is a pre-length space with a local structure provided by localizing neighborhoods and whose time separation function can be recovered from the $\tau$-length of curves. In precise terms, we have the following (see Definitions 3.22 in \cite{KS} and 2.27 in \cite{ACS}).

\begin{definition}\label{defi:lls}
A causally path connected Lorentzian pre-length space for which:
\begin{enumerate}
    \item Every point $x$ has a localizing neighborhood $\Omega_x$.
    \item Every point has a neighborhood in which the relation $\le$ is closed\footnote{This condition is referred to as $(X,d,\ll, \le ,\tau)$ being \textit{locally causally closed.}}.
    \item $\tau (x,y)=\mathcal{T}(x,y)$, for all $x,y\in X$, where
\[
\mathcal{T}(x,y) = \sup\{L_{\tau}(\gamma): \gamma\mbox{ is future-directed causal curve from $x$ to $y$}\}.\footnote{We set $\mathcal{T}(x,y)=0$ when $y\not\in J^+(x)$.}
\]
is called a \emph{Lorentzian length space}.  
\end{enumerate}
\end{definition}

\begin{remark}
Notice that if there is a $\tau$-length realizing causal curve $\gamma $ (that is, $\tau (x,y)=L_\tau (\gamma )$) joining any two causally related points, then condition (3) in Definition \ref{defi:lls} is satisfied. If this is the case, such curves are called \textit{maximal} and we say that  $(X,d,\\ , \le, \tau )$ is a \textit{geodesic} Lorentzian pre-length space\footnote{In Lorentzian geometry, this condition is also known as the Avez-Seifert property}. 
\end{remark}

Examples of Lorentzian length spaces include strongly causal smooth spacetimes \cite{KS}, spacetimes with timelike boundary \cite{ACS},  causally plain spacetimes with $C^0$ metrics \cite{grantc0,hevelin02} and cone structures \cite{mingucone,samann03}.

A causal hierarchy for Lorentzian length spaces can be established in close resemblance to the standard causal hierarchy of smooth spacetimes \cite{ACS}. However, some of the most meaningful levels --like strong causality, non-imprisonment and globally hyperbolicity-- can still be defined for Lorentzian pre-length spaces in general \cite{KS}.

\begin{definition}
A Lorentzian pre-length space $(X,d, \ll ,\leq ,\tau)$ is  \emph{non-totally imprisoning} if there is an upper bound $C(K)$ for all the $d$-lengths of causal curves contained in a compact  set $K$. A non-totally imprisoning Lorentzian pre-length space is \emph{globally hyperbolic} if all its causal diamonds $J^+(p)\cap J^-(q)$ are compact.
\end{definition}

As expected, globally hyperbolic Lorentzian length spaces share in common with their smooth counterparts some remarkable features. For instance, the Avez-Seifert property and the continuity of the time separation \cite{ACS,KS}.

\subsection{Hausdorff distance}

Let $(X,d)$ be a metric space. We define the $d-$\textit{length} of a (continuous) curve $\gamma :[a,b]\to (X,d)$ as 
\[
L_d(\gamma ) =\sup_P\{ S(P)\}
\]
where the supremum is taken over all partitions $P$, $a=t_0<t_1<\cdots <t_{k-1}<t_k=b$, of the closed interval $[a, b]$ and
\[
S(P)=\sum_{i=1}^k d(\gamma (t_{i-1}), \gamma (t_i)) 
\]
is the length of the polygonal curve with vertices on $\gamma$.

A \textit{length space} is a metric space  in which the metric $d$ is induced by the length functional $L_d$. In precise terms,  a metric space $(X, d)$ is a length space if 
\[
d(p, q) = \inf\{L_d(\gamma ) \mid \gamma \textrm{ is a curve joining }p \textrm{ and }q\},\quad \forall p, q\in X.
\]
A curve $\gamma :[a,b]\to (X,d)$ that realizes distance between any pair of its points --that is, if $L_d(\gamma ) = d(p, q)$-- is called a \textit{geodesic segment}. If every pair of points in $(X,d)$ can be joined by a geodesic segment we say that the space is \textit{geodesic.} In such a case we consider all geodesic segments parameterized with respect to arc length. Let us recall that any complete and locally compact length space is geodesic (see for instance Thrm. 2.5.23 \cite{burago}).

The Hausdorff distance $d_H$ of the metric space $(X, d)$ is defined on the set of closed subsets of $(X,d)$ as
\[
d_H(A, B) = \inf \{r \mid A \subset U_r (B), B \subset U_r (A)\}
\]
where $B_r(a)$ is the (open) metric ball of radius $r$ centered at $a$ and  $U_r (A) = \{x \in X \mid \dist (x, A) <r\} = \cup_{a\in A}B_r (a)$ denotes the tubular neighborhood of $A$ of radius $r$ and $\dist (x, A) = \inf \{d(x, a) \mid a\in A\}$. We now list some standard results pertaining tubular neighborhoods and the Hausdorrf distance that are useful for explicit computations

\begin{lemma}\label{TubularCompact}
Let $X$ be a locally compact and complete length space. Then for every compact subset $A$ of $X$ we have:
\begin{enumerate}
\item $\overline{U}_{t}(A)$ is compact for every $t\geq 0$.
\item $\overline{U}_{t}(\overline{U}_{s}(A)) = \overline{U}_{t+s}(A)$ for every $t,s\geq 0$.
\item $d_H(A, B) = \inf \{r \mid A \subset \overline{U}_r (B), B \subset \overline{U}_r (A)\}$,
\item $d_H(A, B) = \max\{ \sup \{ \dist(a, B)\mid {a\in A}\} , \sup \{ \dist(b, A)\mid {b\in B}\} .$
\item There exist midpoints: for every $x,y\in X$ there is $z\in X$ with $d(x,z)=d(y,z)=d(x,y)/2$.
\item Closed $d$-balls are compact. 
\end{enumerate}
\end{lemma}

Given a metric space $(X,d)$, the set 
\[
\mathcal{H}(X)=\{C\subset X: \mbox{$C$ is a compact set of $(X,d)$}\}
\]
is commonly referred to as the \emph{hyperspace of compact sets} of $(X,d)$. When furnished with the Hausdorff distance, its topological structure is one of the main objects of study  in the theory of continua. Here we establish one of its most important geometric aspect, namely,  that $\mathcal{H}(X)$ is indeed geodesic.

\begin{proposition}\label{prop:Hgeo}
Let $X$ be a locally compact and complete length space. For every two compact sets $A,B\in\mathcal{H}(X)$ with $r=d_H(A,B)>0$, the path $\alpha:[0,r]\to \mathcal{H}(X)$ given by
\[
\alpha(t) = \overline{U}_{t}(A)\cap \overline{U}_{r-t}(B)
\]
is a geodesic segment connecting $A$ and $B$.
\end{proposition}

\begin{proof} Since  $r=d_H(A,B)$ then there is a strictly increasing sequence $\{n_{k}\}_{k\in\mathbb{N}}\subset \mathbb{N}$ such that
\[
A\subset \overline{U}_{r+\frac{1}{n_{k}}}(B),
\]
for all $k\in\mathbb{N}$. Thus $A\subset \displaystyle{\bigcap_{k=1}^{\infty} \overline{U}_{r+\frac{1}{n_{k}}}(B) = \overline{U}_{r}(B)}$. From this, $A\subset\overline{U}_r(B)$. Similarly $B\subset\overline{U}_r(A)$. Thus
\[
\alpha(0) = \overline{U}_0(A)\cap \overline{U}_r(B) = A\cap \overline{U}_r(B) = A
\]
and $\alpha(r)=B$.

Now we show that for any $t\in[0,r]$ we have $\alpha(t)\neq\varnothing$. Since $A$, $B$ are compact, there exist $a_0\in A$,  $b_0\in B$ such that $d_H(A,B)=d(a_0,b_0)$. Let $c_0$ be the point in the segment joining $a$ and $b$ that satisfies $d(a_0,c_0)=t$ and $d(c_0,b_0)=r-t$. Thus
\[
\dist(c_0,A)\leq d(c_0,a_0)= t \mbox{ \ and \ } \dist(c_0,B)\leq d(c_0,b_0)= r-t.
\]
Hence $c\in \overline{U}_t(A)$ and $c\in \overline{U}_{r-t}(B)$, therefore $c\in\alpha(t)$.

We now show that $d_H(A,\alpha(t))=t$. Let us consider $a\in A$. Since $d_H(A,B)=r$ there exists $b\in B$ such that $d(a,b)\leq r$. Moreover, there exists  a point $c$ in the segment joining $a$ and $b$ such thta $d(a,c)\leq t$ and $d(c,b)\leq r-t$. Hence
    \[
    \dist(c,A)\leq d(c,a)\leq t \mbox{ \ and \ } \dist(c,B)\leq d(c,b) \leq r-t.
    \]
    and therefore $c\in \alpha(t)$. Hence
        \[
    \dist(a,\alpha(t)) \leq d(a,c)\leq t
    \]
    and $\displaystyle{\sup_{a\in A}\dist(a,\alpha(t))\leq r}$. Now, let $x\in \alpha(t)$, thus $x\in\overline{U}_{t}(A)$ and $\dist(x,A)\leq t$. Hence $\displaystyle{\sup_{x\in\alpha(t)} \dist(x,A)\leq t}$. As a consequence, $d_H(A,\alpha(t))\leq t$. We can show in a similar fashion that $d_H(\alpha(t),B)\leq r-t$. Thus
\[
r=d_H(A,B)\leq d_H(A,\alpha(t)) + d_H(\alpha(t),B) \leq t + (r-t) = r,
\]
and the relations $d_H(A,\alpha(t))=t$ and $d_H(\alpha(t),B)=r-t$
follow. Along the same lines of the previous argument we can show that for all $0\leq s \leq t$ we have
\[
d_H(A,\overline{U}_s(A)\cap \overline{U}_{t-s}(\alpha(t))  ) = s \mbox{ \ and \ } d_{H}( \overline{U}_s(A)\cap \overline{U}_{t-s}(\alpha(t)), \alpha(t))= t-s.
\]
On the other hand, by Lemma \ref{TubularCompact} we also have
\[
\begin{array}{rcl}
 \overline{U}_s(A)\cap\overline{U}_{t-s}(\alpha(t)) &=&  \overline{U}_s(A)\cap \overline{U}_{t-s}(\overline{U}_t(A)\cap \overline{U}_{r-t}(B)) \\
&=& \overline{U}_{s}(A)\cap \overline{U}_{r-s}(B) =\alpha(s).
\end{array}
\]
and hence $d_H(\alpha(s),\alpha(t))=t-s$. 
\end{proof}

\section{Constructions}\label{sec:products}

This section is devoted to general constructions that enable us to build new pre-length spaces from old. The motivation behind these construction is twofold: first, to provide new ways to find examples, and second, to establish a general framework for future applications, most notably, in the study of metric properties of Lorentzian pre-length spaces.

\subsection{Lorentzian taxicab product}

The following construction will be key in examining the space of causal diamonds of Minkowski space. Recall that the taxicab product of the metric spaces $(Y,d_Y)$, $(X,d_X)$ is the metric space $(Y\times X, d_T)$ where for all $a,p\in Y$, $b,q \in X$
\[
d_T((a,b),(p,q)) = d_Y(a,p) + d_X(b,q).
\]

\begin{definition}[Lorentzian taxicab product space] 
Let $(X,d)$ be a metric space and $(Y,d_L,\ll_Y,\leq_Y,\tau_Y)$ a Lorentzian pre-length space. The Lorentzian taxicab product  $Y\times_TX := (Y\times X,d_T,\ll_T,\leq_T,\tau_T)$  is given by 
\begin{itemize}
\item $(a,b)\ll_T (p,q)$ if and only if $\tau_Y(a,p)> d_X(b,q)$.
\item $(a,b)\leq_T (p,q)$ if and only if $\tau_Y(a,p)\geq d_X(b,q)$ and $a\leq_{Y} p$.
\item 
    \begin{equation*}
    \tau_T((a,b),(p,q)) = \left\{
    \begin{array}{cl}
    \tau_Y(a,p) -d_X(b,q)   & \textrm{if } (a,p)\ll_T (b,q)
    \\
    0 & \textrm{otherwise}
    \end{array}
    \right.
    \end{equation*}
\end{itemize}
\end{definition}

\begin{proposition}\label{prop:taxi}
Let $(X,d_X)$ and $(Y,d_Y,\ll_Y,\leq_Y,\tau_Y)$ be a metric space and a Lorentzian pre-length space, respectively. Then the Lorentzian taxicab product $Y\times_T X$ is a Lorentzian pre-length space. 
\end{proposition}

\begin{proof}
First, if $(a,b)\nleq_T (p,q)$ then $\tau_T((a,b),(p,q))=0$. The relation $\leq_T$ is clearly reflexive. For $(a,b)\leq_T (p,q)$ and $(p,q)\leq_T(y,x)$ we have $a\leq_Y p \leq_Y y$, thus 
\[
\tau_Y(a,y) \geq \tau_Y(a,p) + \tau_Y(p,y) \geq d_X(b,q) + d_X(q,x) \geq d_X(b,x),
\]
therefore $(a,b)\leq_T (y,x)$ and $\leq_T$ is a transitive relation. Similarly we can establish that $\ll_T$ is transitive  as well. Now, it follows from the definition that $\tau_T((a,b),(p,q))>0$ if and only if $(a,b)\ll_T (p,q)$. It remains to check that $\tau_T$ is a lower semicontinuous function with respect to $d_T$ and satisfies the reverse triangle inequality. In order to prove the former, let us fix $(a_0,b_0),(p_0,q_0)\in Y\times X$ and $\varepsilon>0$. Since $\tau_Y$ is semicontinuous at $(a_0,p_0)$, there exists $\delta_0>0$ such that if $\displaystyle{\sqrt{d_Y(a,a_0)^2 + d_Y(p,p_0)^2}<\delta_0}$, then $    \tau_Y(a_0,p_0) - \tau_Y(a,p)<\varepsilon /2.$ Set $\delta=\min\left\{{\delta_0}/ {\sqrt{2}}, {\varepsilon}/ {4} \right\}$ and let us take $(a,b),(p,q)\in Y\times X$ such that
    \[
    \displaystyle{\sqrt{d_T((a,b),(a_0,b_0))^2 + d_T((p,q),(p_0,q_0))^2}< \delta.}
    \]
Thus $d_X(b,b_0)\le d_T((a,b),(a_0,b_0))<\delta$ and similarly $d_X(q,q_0)<\delta$, $d_Y(a,a_0)<\delta$ and $d_Y(p,p_0)<\delta$. Hence $\sqrt{d_Y(a,a_0)^2 + d_Y(p,p_0)^2} <\sqrt{2}\delta \leq \delta_0$,
and thus    $\tau_Y(a_0,p_0) - \tau_Y(a,p)<{\varepsilon}/ {2}$.
On the other hand, by the triangle inequality we get
    \[
    d_X(a,p) - d_X(a_0,p_0)\leq d_X(b,b_0) + d_X(q_0,q)<2\delta\leq \displaystyle{\frac{\varepsilon}{2}}.
    \]
In conclusion
    \begin{eqnarray*}
    \left(\tau_Y(a_0,p_0) - d_X(b_0,q_0)\right) &-& \left(\tau_Y(a,p) - d_X(b,q) \right) =\\
    & & \left(\tau_Y(a_0,p_0) - \tau_Y(a,p) \right) + \left(d_X(b,q) -d_X(b_0,q_0) \right) <\varepsilon,
    \end{eqnarray*}
The above inequality readily implies that for $(a,b)\leq_T (p,q)$ we have
\[
 \tau_T((a_0,b_0),(p_0,q_0)) - \tau_T((a,b),(p,q)) <\varepsilon.
\]
On the other hand, in case of $(a,b)\nleq_T (p,q)$ the same statement holds true because in that case $\tau_Y(a,p) - d_X(b,q)<0$ or $a\nleq_Y p$, therefore
    \begin{eqnarray*}
    \tau_T((a_0,b_0),(p_0,q_0)) &-& \tau_T((a,b),(p,q))\\  &\leq& \left(\tau_Y(a_0,p_0) - d_X(b_0,q_0)\right) - \left(\tau_Y(a,p) - d_X(b,q) \right).
    \end{eqnarray*}
Finally, if $(a_0,b_0)\nleq_T (p_0,q_0)$ then
$ -\tau_T((a,b),(p,q))<\varepsilon $ 
 Thus, the proof of semi-continuity is complete. 

Now we tackle the reverse triangle inequality. Let
 $(a,b)\leq_T (p,q) \leq_T (y,x)$, then
    \[
    \begin{array}{rcl}
    \tau_T((a,b),(p,q)) + \tau_T((p,q),(y,x)) &=& \left(\tau_Y(a,p)-d_X(b,q)\right) + \left(\tau_Y(p,y)-d_X(q,x)\right) \\
    &=& \left(\tau_Y(a,p) + \tau_Y(p,y)\right) - \left( d_X(b,q) + d_X(q,x) \right) \\
    &\leq& \tau_Y(a,y) - d_X(b,x) \\
    &=& \tau_T((a,b),(y,x))
    \end{array}
    \]
\end{proof}

Now let us focus on the Lorentzian taxicab product space $\mathbb{R}_{1}^{1}\times_T X$, where $(X,d)$ is a locally compact and complete geodesic length space. Thus we have
\begin{eqnarray*}
(t,x)\ll_T (s,y) &\Leftrightarrow& s-t> d(x,y) \\  (t,x)\leq_T (s,y) &\Leftrightarrow& s-t\geq_T d(x,y)\\
d_T((t,x),(s,y)) &=& \vert t-s\vert +d(x,y), \\
\tau_T((t,x),(s,y))&=&\left\{
\begin{array}{cl}
s-t-d(x,y) & \mbox{if $(t,x)\leq_T (s,y)$} \\
0 & \mbox{otherwise}
\end{array}\right.
\end{eqnarray*}

\begin{remark}\label{ex:mink}
Notice that if $(M,g)$ is a geodesically complete Riemannian manifold then the causal and chronological relations of the Lorentzian length space $\mathbb{R}_1^1\times_T M$ and the standard (Lorentzian) product manifold $-\mathbb{R}\times M$ agree, thus their structures as causal sets are identical.  In particular, their respective causal diamonds coincide as point sets\footnote{Recall that the spacetime $-\mathbb{R}\times M$ is globally hyperbolic if and only $(M,g)$ is complete (see \cite{ON}) }. However, their manifestly different time separation functions give rise to two very different geodesic structures. Indeed, take for instance $M=\mathbb{R}^n$ with its standard Euclidean metric. Then any pair of chronologically related points in Minkowski space $\mathbb{M}^{n+1}=-\mathbb{R}\times \mathbb{R}^n$ can be connected by a \textit{unique} maximal geodesic, namely, the straight line segment joining them. On the other hand, in the Lorentzian taxicab product $\mathbb{R}_1^1\times_T \mathbb{R}^n$ there might be \textit{infinitely many} maximal geodesic segments joining any pair of chronologically related points. Indeed any curve $\gamma :[a,b]\to \mathbb{R}_1^1\times \mathbb{R}^n$ of the form $\gamma (s)= (\alpha (s),\beta (s))$ with $\alpha (s)$ monotone and $\vert\beta (s)\vert$ strictly monotone is maximal.  What is more, there are in $\mathbb{R}_1^1\times \mathbb{R}^n$ causal non-timelike curves which turn out to be maximal \cite{solismontes} (see Figure \ref{fig:rn}). 
\end{remark}

\begin{figure}[ht!]\label{fig:rn}
\centering \includegraphics[scale=0.3]{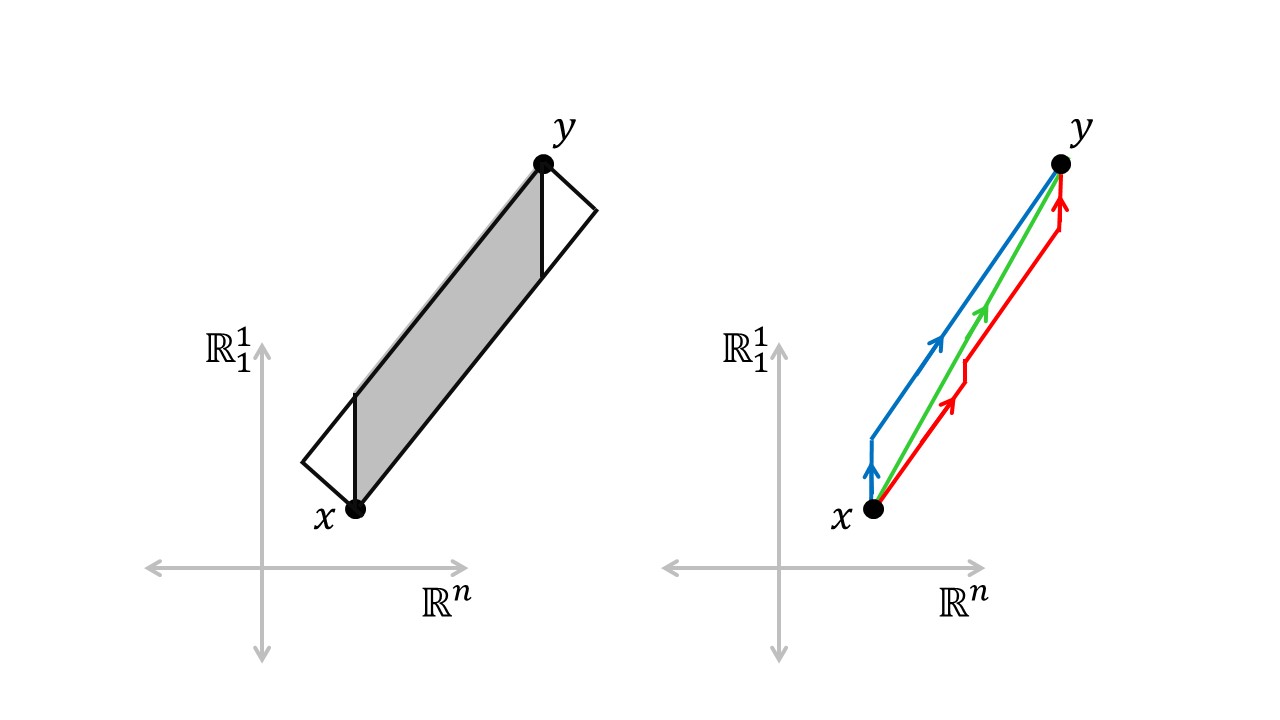}
\centering{\caption{Infinitely many maximal causal curves joining a pair of chronologically related points}} 
\end{figure}

Notice that due to Remark \ref{ex:mink}, the causal diamonds in $\mathbb{R}_{1}^{1}\times_T \mathbb{R}^n$ are compact. In fact, we will show next that $\mathbb{R}_{1}^{1}\times_T X$ is a geodesic Lorentzian length space. Further, if $(X,d)$ is complete then $\mathbb{R}_{1}^{1}\times_T X$ is globally hyperbolic. We start by providing explicit parametrizations for maximal causal curves joining any two causally related points.

\begin{lemma}\label{segmentsRxX}
For every two points $(t,x)$ and $(s,y)$ in $\mathbb{R}_{1}^{1}\times_T X$ with $(t,x)\leq_T (s,y)$ let $L=s-t-d(x,y)$ and $\gamma:[0,d(x,y)]\to X$ be a geodesic segment parametrized by arclength. 
\begin{enumerate}
\item If $L>0$, then $(t,x)\ll_T (s,y)$ and $\ell:[0,L]\to \mathbb{R}^{1}_{1}\times_T X$ defined as
\[
\ell(u)=\left( t+\frac{u(s-t)}{L}, \gamma\left( \frac{u\cdot d(x,y)}{L} \right) \right),
\]
is a maximal future-directed timelike curve from $(t,x)$ to $(s,y)$. 
\item If $L=0$, then $\ell:[0,1]\to \mathbb{R}_{1}^{1}\times_T X$ defined as
\[
\ell(u)=\left( t+u(s-t), \gamma\left( u\cdot d(x,y) \right) \right),
\]
is a maximal future-directed null curve from $(t,x)$ to $(s,y)$.
\end{enumerate}
\end{lemma}

\begin{proof}
First, since $\gamma$ has been parametrized by arclength, then for every $u_1 < u_2$ we have
\[
d_T(\ell(u_1),\ell(u_2))= \frac{(u_2-u_1)(s-t)}{L} +\frac{(u_2-u_1)d(x,y)}{L} = (u_2-u_1)\left(\frac{s-t + d(x,y)}{L}\right),
\]
thus $\ell$ is a non-constant Lipschitz continuous function. In fact, this leads to $L_{d_T}(\ell)=s-t+d(x,y)$, therefore $\ell$ is a geodesic segment connecting $(t,x)$ with $(s,y)$. On the other hand observe that
\[
\tau_T(\ell(u_1),\ell(u_2)) = \frac{(u_2-u_1)(s-t)}{L}-\frac{(u_2-u_1)d(x,y)}{L} = u_2-u_1>0,
\]
thus $\ell(u_1)\ll_T \ell(u_2)$ and the same inequality implies $L_{\tau_T}(\ell)=L$, hence $\ell$ is a future-directed maximal timelike curve. The proof of the null case is analogous.\end{proof}

\begin{theorem}\label{LLSRxX}
The Lorentzian taxicab product space $\mathbb{R}_{1}^{1}\times_T X$ is a globally hyperbolic Lorentzian length space.
\end{theorem}

\begin{proof}
 The space $\mathbb{R}_{1}^{1}\times_T X$ is a geodesic causally path connected Lorentzian pre-length space from Proposition \ref{prop:taxi} and Lemma \ref{segmentsRxX}.
 
Now let $(t,x)\in \mathbb{R}_{1}^{1}\times_T X$, $r>0$ and consider the set
\[
\Omega_{(t,x)}^{r} = I^{+}(t-r,x)\cap I^{-}(t+r,x).
\]
We now show $\Omega_{(t,x)}^{r}$ is a localizing neighborhood around $(t,x)$. Since  $I^{+}(t-r,x)$ and $I^{-}(t+r,x)$ are open, then $\Omega_{(t,x)}^{r}$ is open as well. If $(s_1,y_1),(s_2,y_2)\in \Omega_{(t,x)}^{r}$ satisfy $(s_1,y_1)\le_T (s_2,y_2)$ then the curve $\ell$ constructed in Lemma \ref{segmentsRxX} is a future-directed maximal causal curve from $(s_1,y_1)$ to $(s_2,y_2)$ contained in $ \Omega_{(t,x)}^{r}$. It follows that
\[
(\Omega_{(t,x)},d_T\vert_{\Omega_{(t,x)}\times \Omega_{(t,x)}}, \ll_T\vert_{\Omega_{(t,x)}\times \Omega_{(t,x)}}, \leq_T\vert_{\Omega_{(t,x)}\times \Omega_{(t,x)}}, \tau_T\vert_{\Omega_{(t,x)}\times \Omega_{(t,x)}})
\]
is a Lorentzian pre-length space that fulfills condition (1) and (4) in Definition \ref{defi:lls}. Moreover, let $(s,y)\in \Omega_{(t,x)}^{r}$ and $z$ be a midpoint of $x$ and $y$, hence $d(x,z)=d(y,z)=d(x,y)/2$. Observe that
\[
t+r-\left( \frac{t+r+s}{2} \right) = \frac{t+r-s}{2}>\frac{d(x,y)}{2}= d(x,z),
\]
thus $\left( \frac{t+r+s}{2}, z \right)\ll_T (t+r,x)$. We also have
\[
\left( \frac{t+r+s}{2} \right) - s = \frac{t+r-s}{2}> \frac{d(x,y)}{2} = d(y,z),
\]
then $(s,y)\ll_T \left( \frac{t+r+s}{2}, z \right)$ and therefore $\left( \frac{t+r+s}{2}, z \right) \in I^{+}(s,y)\cap \Omega_{(t,x)}^{r}\neq \varnothing$. A similar argument shows that  $I^{-}(s,y)\cap \Omega_{(t,x)}^{r}\neq \varnothing$, and hence condition (3) in Definition \ref{defi:lls} holds as well.

Now, fix a future-directed causal curve $\gamma:[a,b]\to \mathbb{R}_{1}^{1}\times_T X$ with $\textrm{Im}(\gamma)\subset \Omega_{(t,x)}^{r}$. The curve $\gamma$ can be written as $\gamma(u)=(\alpha(u),\beta(u))$ where $\alpha:[a,b]\to \mathbb{R}$ and $\beta:[a,b]\to X$. Since $\gamma$ is causal then $\alpha(u_1)\leq_T \alpha(u_2)$ and $\alpha(u_2)-\alpha(u_1)\geq d(\beta(u_1),\beta(u_2))$ for every $u_1\leq u_2$. In particular $(t-r,x)\ll_T \gamma(a) \leq_T \gamma(b)\ll_T (t+r)$ leads to $t+r-\alpha(b)>0$ and $\alpha(a)+r-t>0$. Take a partition $a=u_0<u_1<\cdots < u_{n}=b$, thus

\begin{eqnarray*}
\displaystyle{\sum_{i=0}^{n-1} d_T(\gamma(u_i),\gamma(u_{i+1}))} &=& \displaystyle{\sum_{i=0}^{n-1} \left[\alpha(u_{i+1})-\alpha(u_{i}) + d(\beta(u_{i}),\beta(u_{i+1})) \right]} \\
&\leq & 2 \displaystyle{\sum_{i=0}^{n-1} \left[\alpha(u_{i+1})-\alpha(u_{i}) \right]} \\
&=& 2\left[ \alpha(b)-\alpha(a) \right] <(t+r)+(r-t)= 2r.
\end{eqnarray*}

In conclusion, the length of any causal curve contained in $\Omega_{(t,x)}^{r}$ satisfies $L_{d_T}(\gamma)<2r$ and thus condition (2) in Definition \ref{defi:lls} holds as well. Hence  $\Omega_{(t,x)}^{r}$ is localizing.

Finally, take two convergent sequences $(t_n,x_n)\to (t,x)$ and $(s_n,y_n)\to (s,y)$ with $(t_n,x_n)\leq_T (s_n,y_n)$ for all $n\in\mathbb{N}$. Since $d_T((t_n,x_n),(t,x))=\vert t_n-t\vert +d(x,x_n)$ it readily follows that $t_n\to t$, $x_n\to x$, and similarly $s_n\to s$,  $y_n\to y$. Thus $s_n-t_n\geq d(x_n,y_n)$ implies $s-t\geq d(x,y)$ and $(t,x)\leq_T (s,y)$, hence $\mathbb{R}_{1}^{1}\times_T X$ is locally causally closed.  Thus $\mathbb{R}^1_1\times_T X$ is a geodesic Lorentzian length space.

Now we prove $\mathbb{R}^1_1\times_T X$ is globally hyperbolic. First, if $K \subset\mathbb{R}_{1}^{1}\times_T X$  compact then there exists a number $0<C<\infty$ such that
\[
0\leq 2\sup\{s_0-t_0: (t_0,x_0),(s_0,y_0)\in K\}<C.
\]
Let $\gamma:[a,b]\to \mathbb{R}_{1}^{1}\times_T X$ be a future-directed causal curve contained in $K$, then 
\[
L_{d_T}(\gamma)\leq 2\left[ \alpha(b) - \alpha(a)\right]  <C.
\]
Hence $\mathbb{R}_{1}^{1}\times_T X$ is non-totally imprisoning. Finally, we prove that $D=J^{+}(t,x)\cap J^{-}(s,y)$ is compact for every $(t,x)\leq_T (s,y)$. Let $\{(u_n,z_n)\}_{n\in\mathbb{N}}$ be a sequence in $D$, then $t\leq u_n \leq s$ and $s-t\geq d(x,z_n)+d(y,z_n)\geq d(x,z_n)$, then the sequence $\{z_n\}_{n\in\mathbb{N}}$ is contained in a compact closed ball in $X$. Hence we can choose a subsequence $\{(u_{n_{k}},z_{n_{k}})\}_{k\in \mathbb{N}}$ --using a Cantor's diagonal argument-- with $(u_{n_{k}},z_{n_{k}})\to (u,z)$. Since $s-u_{n_{k}}\geq d(y,z_{n_{k}})$ and $u_{n_{k}}-t\geq d(x,z_{n_{k}})$, by taking $k\to\infty$ we conclude $(u,z)\in D$. Therefore $D$ is compact and the proof is complete.
\end{proof}

\subsection{Lorentzian uniform product}

In analogy with the $L_\infty$ norm, given two metric spaces $(X,d_X)$ and $(Y,d_Y)$ the \textit{uniform metric} $d_{\infty}:(X\times Y)\times (X\times Y)\to \mathbb{R}$ is defined by  
\[
d_{\infty}((x,y),(a,b)) = \max\{d_{X}(x,a),d_{Y}(y,b)\}.
\]

\begin{definition}[Lorentzian uniform product space]
Let $(X,d_{X},\ll_{X},\leq_{X},\tau_{X})$ and $(Y,d_{Y},\ll_{Y},\leq_{Y},\tau_{Y})$ be two Lorentzian pre-length spaces. The Lorentzian uniform product space $X\times_\infty Y=(X\times Y,d_{\infty},\ll_{\infty},\leq_{\infty},\tau_{\infty})$ is defined by
\begin{itemize}
\item $(x,y)\ll_{\infty} (a,b)$ if and only if $x\ll_{X}a$ and $y\ll_{Y}b$.
\item $(x,y)\leq_{\infty} (a,b)$ if and only if $x\leq_{X}a$ and $y\leq_{Y}b$.
\item 
\[
\tau_\infty((x,y),(a,b)) =   \min\{\tau_{X}(x,a),\tau_{Y}(y,b)\}.
\]
\end{itemize}
\end{definition}

\begin{proposition}
The Lorentzian uniform product $X\times_\infty Y$ is a Lorentzian pre-length space.
\end{proposition}

\begin{proof} It is not hard to see that the relations  $\ll_\infty$, $\leq_\infty$ satisfy the required properties of a pre-length space since both $\ll_{X}$, $\leq_{X}$ and $\ll_{Y}$, $\leq_{Y}$ do. If $(x,y)\nleq_{\infty}(y,b)$ then $x\nleq_{X}a$ or $y\nleq_{Y}b$ and therefore $\tau_{X}(x,a)=0$ or $\tau_{Y}(y,b)=0$ leading to $\tau_{\infty}((x,y),(a,b))=0$. In fact $(x,y)\ll_{\infty}(a,b)$ if and only if $\tau_{X}(x,a)>0$ and $\tau_{Y}(y,b)>0$, which in turn holds if and only if $\tau_{\infty}((x,y),(a,b))>0$. Thus, it only remains to verify the reverse triangle inequality and lower semicontinuity of $\tau_\infty$ with respect to $d_\infty$.

We first analyze the lower semicontinuity of $\tau_{\infty}$.  Let $\varepsilon>0$ and fix $((x_0,y_0),(a_0,b_0))\in (X\times Y)\times (X\times Y)$. Since $\tau_{X}$ is lower semicontinuous at $(x_0,a_0)$ with respect to $d_{X}$, we can find $\delta_{X}>0$ such that
 $\displaystyle{\sqrt{d_{X}^{2}(x,x_0)+d_{X}^{2}(a,a_0)}<\delta_{X}}$ implies $\tau_{X}(x_0,a_0)-\tau_{X}(x,a)<\varepsilon$. Choose $\delta=\min\{\delta_{X} / \sqrt{2},\delta_{Y} / \sqrt{2}\}$ and suppose
\[
\displaystyle{\sqrt{d_{\infty}^{2}((x,y),(x_0,y_0)) + d_{\infty}^{2}((a,b),(a_0,b_0))}<\delta},
\]  
implying $d_X(x,x_0)<\delta$ and $d_X(a,a_0)<\delta$. Thus
\[
\displaystyle{\sqrt{d_{X}^{2}(x,x_0)+d_{X}^{2}(a,a_0)}<\sqrt{2}\delta}\leq \delta_{X},
\]
leading to $\tau_{X}(x_0,a_0)-\tau_{X}(x,a)<\varepsilon$. Hence
\[
\tau_{\infty}((x_0,y_0),(a_0,b_0))-\tau_{X}(x,a) \leq \tau_{X}(x_0,a_0)-\tau_{X}(x,a)<\varepsilon ,
\]
therefore
\[
\tau_{\infty}((x_0,y_0),(a_0,b_0))-\tau_{X}(x,a)<\varepsilon .
\]
A similar argument enables us to show
\[
\tau_{\infty}((x_0,y_0),(a_0,b_0))-\tau_{Y}(y,b) <\varepsilon ,
\]
so in conclusion
\[
\tau_{\infty}((x_0,y_0),(a_0,b_0))-\tau_{\infty}((x,y),(a,b))<\varepsilon
\]
thus establishing the lower semicontinuity of $\tau_\infty$ at 
$((x_0,y_0),(a_0,b_0))$.

 Finally, if $(x_1,y_1)\leq_{\infty} (x_2,y_2)\leq_{\infty} (x_3,y_3)$ then $x_1\leq_{X} x_2 \leq_{X} x_3$ and $y_1\leq_{Y} y_2\leq_{Y} y_3$. Then by using the reverse triangle inequality for $\tau_{X}$ and $\tau_{Y}$ we obtain
\begin{eqnarray*}
\tau_{X}(x_1,x_3)&\geq& \tau_{X}(x_1,x_2) + \tau_{X}(x_2,x_3) \\ &\geq& \tau_{\infty}((x_1,y_1),(x_2,y_2)) + \tau_{\infty}((x_2,y_2),(x_3,y_3)),
\end{eqnarray*}
and similarly
\[
\tau_{Y}(y_1,y_3)\geq  \tau_{\infty}((x_1,y_1),(x_2,y_2)) + \tau_{\infty}((x_2,y_2),(x_3,y_3)).
\]
Hence
\[
\tau_{\infty}((x_1,y_1),(x_3,y_3))\geq \tau_{\infty}((x_1,y_1),(x_2,y_2)) + \tau_{\infty}((x_2,y_2),(x_3,y_3))
\]
which finishes the proof of the reverse triangle inequality.
\end{proof}

Now we move on in to proving that under some mild assumptions, the uniform product of two Lorentzian length spaces is a Lorentzian length space. A few previous lemmas are in order. First we focus on the geodesic structure of the uniform product.

\begin{remark}\label{CausalConnectness}
Let $\gamma:[a,b]\to X\times_{\infty} Y$, $\gamma(t)=(\gamma_X(t), \gamma_Y(t))$ be a curve.  First, notice that if $B_{r}(x)$ and $B_{r}(y)$ are balls in $(X,d_X)$ and $(Y,d_Y)$, respectively, then the $d_\infty$ ball $B_{r}(x,y)$ agrees with $B_{r}(x)\times B_{r}(y)$. We immediately conclude that if $\gamma$ is locally Lipschitz continuous then so are $\gamma_{X}$ and $\gamma_{Y}$. Furthermore, if $\gamma$ is a future-directed causal curve, then for all $a\leq s<t \leq b$ we have $(\gamma_X(s),\gamma_Y(s)) \leq_{\infty} (\gamma_X(t),\gamma_Y(t))$ implying $\gamma_X(s)\leq_{X}\gamma_X(t)$ and $\gamma_Y(s)\leq_{S}\gamma_Y(t)$. If $\gamma_X$ is not constant then $\gamma_X$ is a future-directed causal curve. Moreover, if $\gamma$ is timelike or null then so is $\gamma_{X}$.
\end{remark}

\begin{lemma}\label{UniformCausalCurves}
Let $\gamma_1:[a_1,b_1]\to X$ and $\gamma_2:[a_2,b_2]\to Y$ be two future causal curves in $(X,d_X,\ll_X,\leq_X,\tau_X)$ and $(Y,d_Y,\ll_Y,\leq_Y,\tau_Y)$, respectively. Let us consider $\varphi_i:[0,1]\to [a_i,b_i]$ for $i=1,2$ given by $\varphi_i(t) = t(b_i-a_i)+a_i$. Then $\gamma:[0,1]\to X\times_{\infty} Y$ defined as $\gamma(t)=(\gamma_1(\varphi_1(t)), \gamma_2(\varphi_2(t)))$ is a causal curve in $X\times_{\infty} Y$ and 
\[
L_{\tau_{\infty}}(\gamma)\leq \min\{L_{\tau_X}(\gamma_1), L_{\tau_{Y}}(\gamma_2)\}.
\]
Moreover, if $\gamma_1$ and $\gamma_2$ are timelike and maximal, then $\gamma$ is maximal and timelike
\end{lemma}

\begin{proof}
Note that $\gamma$ is future causal since $\gamma_1$ and $\gamma_2$ also are, and  $\varphi_1$, $\varphi_2$ are continuous strictly increasing functions. To analyze the $\tau_\infty$ length of curves, fix two partitions $a_1=t_0<t_1<\cdots <t_n = b_1$ and $a_2=s_0<s_1<\cdots< t_m=b_2$ in $[a_1,b_1]$ and $[a_2,b_2]$, respectively. Thus we induce two partitions  $0=\varphi_1^{-1}(a_1)=\varphi_1^{-1}(t_0)<\varphi_1^{-1}(t_1)<\cdots <\varphi_1^{-1}(t_n) = \varphi_1^{-1}(b_1)=1$ and $0=\varphi_2^{-1}(a_2)=\varphi_2^{-1}(s_0)<\varphi_2^{-1}(s_1)<\cdots< \varphi_2^{-1}(t_m)=\varphi_2^{-1}(b_2)=1$ in $[0,1]$. Now
\[
\begin{array}{rcl}
L_{\tau_{\infty}}(\gamma) &\leq& \displaystyle{\sum_{i=0}^{n-1} \tau_{\infty}(\gamma(\varphi_1^{-1}(t_{i+1})), \gamma(\varphi_1^{-1}(t_{i+1})))} \\
&=& \displaystyle{\sum_{i=0}^{n-1} \min\left\{\tau_X(\gamma_1(t_{i}), \gamma_1(t_{i+1})), \tau_Y(\gamma_2(\varphi_1^{-1}(t_{i})), \gamma_2(\varphi_1^{-1}(t_{i+1})))  \right\}} \\
&\leq& \displaystyle{\sum_{i=0}^{n-1} \tau_X(\gamma_1(t_{i}), \gamma_1(t_{i+1}))}.
\end{array}
\]
then $L_{\tau_{\infty}}(\gamma)\leq L_{\tau_{X}}(\gamma_1)$, and in a similar way we obtain $L_{\tau_{\infty}}(\gamma)\leq L_{\tau_{Y}}(\gamma_2)$. Therefore 
\[
L_{\tau_{\infty}}(\gamma)\leq \min\{L_{\tau_X}(\gamma_1), L_{\tau_{Y}}(\gamma_2)\}.
\]

Now suppose $\gamma_1$ and $\gamma_2$ are both timelike and maximal and $r_1=\tau_X(\gamma_1(a_1),\gamma_1(b_1))$, $r_2=\tau_Y(\gamma_2(a_2),\gamma_2(b_2))$.  Then $a_i$ and $b_i$ can be chosen such that $r_i=b_i-a_i>0$ and $\gamma_i$ is parametrized by arclength for $i=1,2$. Therefore, for a given partition $0=u_0<u_1<\cdots < u_n=1$ we have
\[
\begin{array}{rcl}
\tau_{\infty}(\gamma(u_{i}),\gamma(u_{i+1})) &=&  \min\{ \tau_{X}(\gamma_1(\varphi_1(u_i))), \gamma_1(\varphi_1(u_{i+1}))), \tau_{X}(\gamma_2(\varphi_2(u_i))), \gamma_2(\varphi_2(u_{i+1}))) \} \\
&=& \min\{ \varphi_1(u_{i+1})-\varphi_1(u_{i}), \varphi_2(u_{i+1})-\varphi_2(u_{i}) \}.
\end{array}
\]
In conclusion
\[
\begin{array}{rcl}
\displaystyle{\sum_{i=0}^{n-1} \tau_{\infty}(\gamma(u_{i}),\gamma(u_{i+1}))} &=& \displaystyle{\sum_{i=0}^{n-1} \min\{ \varphi_1(u_{i+1})-\varphi_1(u_{i})}, \varphi_2(u_{i+1})-\varphi_2(u_{i}) \} \\
&=& \displaystyle{\sum_{i=0}^{n-1} \min\{ (u_{i+1}-u_{i})(b_1-a_1), (u_{i+1}-u_{i})(b_2-a_2) \} }\\
&=& \min\{r_1,r_2\} \displaystyle{\sum_{i=0}^{n-1} u_{i+1}-u_{i} }= \min\{r_1,r_2\}.
\end{array}
\]
Thus $L_{\tau_{\infty}}(\gamma)=\tau_{\infty}(\gamma(0),\gamma(1))$ and therefore $\gamma$ is maximal.
\end{proof}

\begin{lemma}\label{lemma:XYU}
Let $(X,d_{X},\ll_{X},\leq_{X},\tau_{X})$ and $(Y,d_{Y},\ll_{Y},\leq_{Y},\tau_{Y})$ be two Lorentzian pre-length spaces.
\begin{enumerate}
\item[(i)] If $(X,d_{X},\ll_{X},\leq_{X},\tau_{X})$ and $(Y,d_{Y},\ll_{Y},\leq_{Y},\tau_{Y})$ are locally causally closed, then $X\times_{\infty} Y$ is locally causally closed.
\item[(ii)] If $(X,d_{X},\ll_{X},\leq_{X},\tau_{X})$ and $(Y,d_{Y},\ll_{Y},\leq_{Y},\tau_{Y})$ are causally path connected, then $X\times_{\infty} Y$ is causally path connected.
\item[(iii)] If $(X,d_{X},\ll_{X},\leq_{X},\tau_{X})$ and $(Y,d_{Y},\ll_{Y},\leq_{Y},\tau_{Y})$ are regularly localizable, then $X\times_{\infty} Y$ is localizable.
\item[iv)] If there exists timelike maximal curves joining any pair of chronologically related points in both $(X,d_{X},\ll_{X},\leq_{X},\tau_{X})$ and $(Y,d_{Y},\ll_{Y},\leq_{Y},\tau_{Y})$, then $X\times_\infty Y$ is intrinsic.
\end{enumerate}
\end{lemma}

\begin{proof}
For (i) observe that the product $U_{x}\times U_{y}$ of two causally closed neighborhoods $U_x$, $U_y$ containing $x\in X$, $y\in Y$ is a causally closed neighborhood in $X\times Y$. Notice also that $\overline{U_{x} \times U_{y}} = \overline{U}_{x}\times \overline{U}_{y}$. Part (ii) follows from Remark \ref{CausalConnectness}. Part (iv) is a immediate consequence of Lemma \ref{UniformCausalCurves}. Thus we will focus on Part (iii). For every $x\in X$ and $y\in Y$ choose localizing neighborhoods $\Omega_{x}$ and $\Omega_{y}$ and their corresponding continuous maps $\omega_{x}:\Omega_{x}\times \Omega_{x}\to [0,\infty)$, $\omega_{y}:\Omega_{y}\times \Omega_{y}\to [0,\infty)$. We will prove that $\Omega_{(x,y)}$ is a localizing neighborhood and $\omega_{(x,y)}:\Omega_{(x,y)}\times \Omega_{(x,y)}\to [0,\infty)$ defined as
\[
\omega_{(x,y)}((p,q),(z,w)) = \min\{\omega_{x}(p,z), \omega_{y}(q,w)\},
\]
satisfies the conditions of Definition 3.16 of \cite{KS}. Let $\gamma:[a,b]\to X\times_\infty Y$ be a causal curve contained in $\Omega_{(x,y)}$, with $\gamma(t)=(\alpha(t),\beta(t))$. Since $\gamma$ cannot be constant then  $\alpha$ and $\beta$ are not simultaneously constant. Since $\textrm{Im}(\alpha)\subset\Omega_{x}$ and $\textrm{Im}(\beta)\subset\Omega_{y}$ then there exist $C_{1},C_2>0$ with $L_{d_{X}}(\alpha)\leq C_1$ and $L_{d_{Y}}(\beta)\leq C_2$. Moreover, for a given partition $a=t_0<t_1<\cdots<t_n=b$ we have
\[
\begin{array}{rcl}
\displaystyle{\sum_{i=0}^{n-1} d_{\infty}(\gamma(t_{i}),\gamma(t_{i+1}))} &=& \displaystyle{\sum_{i=0}^{n-1} \max\{d_{X}(\alpha(t_{i}),\alpha(t_{i+1})), d_{Y}(\beta(t_{i}), \beta(t_{i+1}))\} } \\
&=& \displaystyle{\sum_{i=0}^{n-1} d_{X}(\alpha(t_{i}),\alpha(t_{i+1})) + d_{Y}(\beta(t_{i}), \beta(t_{i+1}))} \\
&=& \displaystyle{\sum_{i=0}^{n-1} d_{X}(\alpha(t_{i}),\alpha(t_{i+1})) + \sum_{i=0}^{n-1} d_{Y}(\beta(t_{i}), \beta(t_{i+1}))} \\
&\leq& L_{d_{X}}(\alpha) + L_{d_{Y}}(\beta) \\
&\leq& C_1 + C_2.
\end{array}
\]
Then $L_{d_{\infty}}(\gamma)\leq C_1+C_2$.  Moreover, clearly $(\Omega_{(x,y)},d_{\infty},\ll_{\infty}, \leq_{\infty}, \omega_{(x,y)})$ is a Lorentzian pre-length space. Now, if $(p,q)\in \Omega_{(x,y)}$, then 
\[
\varnothing \neq \left(I^{+}_{X}(p)\cap \Omega_{x}\right) \times \left(I^{+}_{Y}(q)\cap \Omega_{y}\right) \subset I^{+}_{\infty}(p,q)\cap \Omega_{(x,y)},
\]
thus $I^{+}_{\infty}(p,q)\cap \Omega_{(x,y)}\neq \varnothing$. Similarly $I^{-}_{\infty}(p,q)\cap \Omega_{(x,y)}\neq \varnothing$. Finally fix $(x_1,y_1),(x_2,y_2)\in \Omega_{(x,y)}$, then there exist two timelike maximal curves $\gamma_1:[a_1,b_1]\to \Omega_{x}$ and $\gamma_2:[a_2,b_2]\to \Omega_{y}$ with $\gamma_1(a_1)=x_1$, $\gamma_1(b_1)=x_2$, $\gamma_{2}(a_2)=y_1$, $\gamma_{2}(b_2)=y_2$ and
\[
L_{\tau_{X}}(\gamma_1)= \omega_{x}(x_1,x_2) \leq \tau_{X}(x_1,x_2),
\]
\[
L_{\tau_{Y}}(\gamma_2)= \omega_{y}(y_1,y_2) \leq \tau_{Y}(y_1,y_2).
\]
In conclusion, the curve $\gamma$ defined as in Lemma \ref{UniformCausalCurves} satisfy
\[
\begin{array}{rcl}
L_{\tau_{\infty}}(\gamma) &=& \min\{L_{\tau_{X}}(\gamma_1), L_{\tau_{Y}}(\gamma_2)\} \\
&=& \min\{\omega_{x}(x_1,x_2),\omega_{y}(y_1,y_2)\} \\
&\leq& \min\{\tau_{X}(x_1,x_2), \tau_{Y}(y_1,y_2)\}.
\end{array}
\]
\end{proof}

\begin{proposition}\label{propo:GHI}
Let $(X,d_{X},\ll_{X},\leq_{X},\tau_{X})$ and $(Y,d_{Y},\ll_{Y},\leq_{Y},\tau_{Y})$ be two globally hyperbolic Lorentzian pre-length spaces. Then $X\times_{\infty} Y$ is a globally hyperbolic Lorentzian pre-length space.
\end{proposition}

\begin{proof}
If $K\subset X\times_{\infty} Y$ is compact, then there exist two compact sets $K_1\subset X$ and $K_2\subset Y$ with $K\subset K_1\times K_2$. Let $\gamma:[a,b]\to K$ to be a causal curve in $X\times_{\infty} Y$, then we can find two causal curves $\alpha:[a,b]\to K_1$, $\beta:[a,b]\to K_2$ with $\gamma(t)=(\alpha(t),\beta(t))$. Since $X$ and $Y$ are globally hyperbolic then there are two positive numbers $C_1$ and $C_2$ with $L_{d_{X}}(\alpha)\leq C_1$ and $L_{d_{Y}}(\beta)\leq C_2$. Therefore, by following a similar argument as in the proof of Lemma  \ref{lemma:XYU} we conclude
\[
L_{d_{\infty}}(\gamma)\leq C_1+C_2.
\]
Thus, $X\times_{\infty} Y$ is non totally imprisoning. Finally, observe that
\[
J^{+}_{\infty}(x_1,y_1)\cap J^{-}_{\infty}(x_2,y_2) = \left(J^{+}_{X}(x_1)\cap J^{-}_{X}(x_2)\right) \times \left(J^{+}_{Y}(y_1)\cap J^{-}_{Y}(y_2)\right),
\]
thus the compactness of $J^{+}_{\infty}(x_1,y_1)\cap J^{-}_{\infty}(x_2,y_2)$ is a direct consequence of the compactness of $J^{+}_{X}(x_1)\cap J^{-}_{X}(x_2)$ and $J^{+}_{Y}(y_1)\cap J^{-}_{Y}(y_2)$. 
\end{proof}

From Lemma \ref{lemma:XYU} and Proposition \ref{propo:GHI} the following result is immediate.

\begin{corollary}
Let $(X,d_{X},\ll_{X},\leq_{X},\tau_{X})$ and $(Y,d_{Y},\ll_{Y},\leq_{Y},\tau_{Y})$ be two globally hyperbolic regularly localizable Lorentzian pre-length spaces. Then $X\times_{\infty} Y$ is a globally hyperbolic Lorentzian length space.

\end{corollary}


\section{Hyperspaces as Lorentzian pre-Length spaces}\label{sec:diamonds}

In this section we establish the main results of this work, namely, we endow with a Lorentzian pre-length structure the hyperspaces of compact sets and causal diamonds of a Lorentzian pre-length space  $(X, d,\ll,\leq,\tau)$ with a continuous time separation $\tau$.

In the context of Lorentzian pre-length spaces, we can define relations $\ll_H$ and $\leq_H$ that naturally inherit the properties of a chronological and causal relation, respectively, from those of $\ll$ and $\leq$. In precise terms we define
\begin{enumerate}
\item[(i)] For $A,B\in \mathcal{H}(X)$, $A\ll_H B$ if and only if $\forall a\in A$ there is a $b_0\in B$ with $a\ll b_0$ and $\forall b\in B$ there is a $a_0\in A$ with $a_0\ll b$.
\item[(ii)] For $A,B\in \mathcal{H}(X)$, $A\leq_H B$ if and only if $\forall a\in A$ there is a $b_0\in B$ with $a\leq b_0$ and $\forall b\in B$ there is a $a_0\in A$ with $a_0\leq b$.
\end{enumerate}

Moreover, for $x\in X$ and $C\in\mathcal{H}(X)$ we define $\dist^{\pm }_{L}:X\times\mathcal{H}(X) \to [0,\infty]$ as
\begin{eqnarray*}
\dist^{-}_{L}(x,C) &=& \sup\{\tau(x,c): c\in C\} ,\\ \dist^{+}_{L}(C,x) &=& \sup\{\tau(c,x):c\in C\},
\end{eqnarray*}
and  $\tau_{H}:\mathcal{H}(X)\times \mathcal{H}(X)\to [0,\infty]$ as\footnote{In case $\{r: A\subset\overline{U}^{L,-}_{r}(B)$  and $B\subset\overline{U}^{L,+}_{r}(A) \} = \varnothing$, then we set $\tau_{H}(A,B)=\infty$. } 
\[
\tau_{H}(A,B) = \sup\{r>0: A\subset\overline{U}^{L,-}_{r}(B) \ \mbox{and} \ B\subset\overline{U}^{L,+}_{r}(A) \},
\]
where 
\begin{eqnarray*}
\overline{U}^{L,-}_{r}(C) &=&\{x\in X: \dist^{-}_{L}(x,C)\geq r\},\\
\overline{U}^{L,+}_{r}(C)&=&\{x\in X: \dist^{+}_{L}(C,x)\geq r\}. 
\end{eqnarray*}

\begin{definition}
We say that $(\mathcal{H}(X), d_H ,\ll_{H},\leq_{H},\tau_{H})$ is the \emph{Lorentzian hyperspace of compact subsets} of $(X,d)$.
\end{definition}

\begin{remark} \label{remarkCont}
Due to the continuity of $\tau$ we have that $\dist^{-}_{L}$ and $\dist^{+}_{L}$ are both non-negative continuous functions. Moreover, for every $x\in X$ and $C\in\mathcal{H}(X)$ there exists $c_{-},c_+\in C$ such that $\dist^{-}_{L}(x,C)=\tau(x,c_{-})$ and  $\dist^{+}_{L}=\tau(x,c_{+})$.  
\end{remark}

The following Lemma can be interpreted as analogues for $\tau_H$ of well known results pertaining to the Hausdorff distance (see for instance Lemma \ref{TubularCompact}).

\begin{lemma}\label{TauH}
For every $A,B\in\mathcal{H}(X)$, $r_1>0$, $r_2>0$ and $r\geq 0$ we have:
\begin{enumerate}
\item $\overline{U}^{L,-}_{r_1}\left(\overline{U}^{L,-}_{r_2}(A)\right)\subset \overline{U}^{L,-}_{r_1+r_2}(A)$ and $\overline{U}^{L,+}_{r_1}\left(\overline{U}^{L,+}_{r_2}(A)\right)\subset \overline{U}^{L,+}_{r_1+r_2}(A)$.
\item $\tau_{H}(A,B)=\min\left\{\displaystyle{\inf_{a\in A}\dist^{-}_{L}(a,B)},\displaystyle{\inf_{b\in B}\dist^{+}_{L}(b,A)}\right\}$.
\item $\tau_{H}(A,B)\geq r$ if and only if $\forall a\in A$ $\dist^{-}_{L}(a,B)\geq r$ and $\forall b\in B$ $\dist^{+}_{L}(A,b)\geq r$.
\item If $A\subset B$ then $\overline{U}^{L,-}_{r}(A)\subset \overline{U}^{L,-}_{r}(B)$ and $\overline{U}^{L,+}_{r}(A)\subset \overline{U}^{L,+}_{r}(B)$.
\end{enumerate}
\end{lemma}

\begin{proof}
In order to prove (1), consider $x\in \overline{U}^{L}_{r_1}\left(\overline{U}^{L,-}_{r_2}(A)\right)$, then there exist $y\in \overline{U}^{L,-}_{r_2}(A)$ and $z \in A$ such that
\[
\tau(x,y) = \dist^{-}_{L}(x,\overline{U}^{L}_{r_2}(A))\geq r_1
\]
\[
\tau(y,z) = \dist_{L}(y,A)\geq r_2,
\]
therefore $\tau(x,y)>0$ and $\tau(y,z)>0$, implying $x\ll y \ll z$ and hence $\tau(x,z)\geq \tau(x,y)+\tau(y,z)$. Thus
\[
\dist^{-}_{L}(x,A)\geq \tau(x,z)\geq r_1+r_2.
\] 
A similar computation shows that $\overline{U}^{L,+}_{r_1}\left(\overline{U}^{L,+}_{r_2}(A)\right)\subset \overline{U}^{L,+}_{r_1+r_2}(A)$.

For point (2), let us denote
\[
R = \min\left\{\displaystyle{\inf_{a\in A}\dist^{-}_{L}(a,B)},\displaystyle{\inf_{b\in B}\dist^{+}_{L}(A,b)}\right\}.
\]
First we will show  $R\geq \tau_{H}(A,B)$. Let $r>0$ such that $A\subset \overline{U}^{L,-}_{r}(B)$ and $B\subset \overline{U}^{L,+}_{r}(A)$. Then $\forall a\in A$ and $\forall b\in B$, $\dist^{-}_{L}(a,B)\geq r$ and $\dist^{+}_{L}(A,b)\geq r$. Therefore
\[
\displaystyle{\inf_{a\in A} \dist^{-}_{L}(a,B)\geq r} \mbox{\ and \ } \displaystyle{\inf_{b\in B} \dist^{+}_{L}(A,b)\geq r}.
\]
Thus $R\geq r$ which readily implies $R\geq \tau_{H}(A,B)$. Conversely, we want to show  $\tau_H(A,B)\geq R$. If $R=0$ then we are done, so suppose $R>0$. Since
\[
\displaystyle{\inf_{a\in A} \dist^{-}_{L}(a,B)\geq R} \mbox{\ and \ } \displaystyle{\inf_{b\in B} \dist^{+}_{L}(A,b)\geq R}.
\]
Thus $\dist^{-}_{L}(a,B)\geq R$ and $\dist^{+}_{L}(A,b)\geq R$, $\forall a\in A$ and $\forall b\in B$, which means $A\subset \overline{U}^{L,-}_{R}(B)$ and $B\subset \overline{U}^{L,+}_{R}(A)$, hence
\[
R\in \{r>0: A\subset\overline{U}^{L,-}_{r}(B) \ \mbox{and} \ B\subset\overline{U}^{L,+}_{r}(A)\}.
\]
In conclusion, $R\leq \tau_{H}(A,B)$ and consequently $R= \tau_{H}(A,B)$.

Point (3) is a direct consequence of point (2). Finally point (4) is derived from the definition of $\dist^{-}_{L}$ and $\dist^{+}_{L}$.
\end{proof}

\begin{lemma} \label{tauHmayorcero}
If $\tau_{H}(A,B)=0$ then $A{\not\ll_{H}} B$.
\end{lemma}

\begin{proof}
If $\tau_{H}(A,B)=0$ then
\[
\displaystyle{\inf_{a\in A}\dist_{L}^-(a,B) = 0} \mbox{ \ or  \ }\displaystyle{\inf_{b\in B}\dist_{L}^+(A,b)=0}.
\]
Suppose $\displaystyle{\inf_{a\in A}\dist^{-}_{L}(a,B) = 0}$. Then there exists a sequence of points $\{x_i\}_{i\in\mathbb{N}}\subset A$ with $\dist^{-}_{L}(x_i,B)\to 0$ when $i\to \infty$. Since $A$ is compact, then there is a convergent subsequence of $\{x_i\}_{i\in\mathbb{N}}$, and without loss of generality assume that $\{x_i\}_{i\in\mathbb{N}}$ converges itself to a point $x\in A$. Note
\[
\dist^{-}_{L}(x,B) = \displaystyle{\lim_{i\to\infty} \dist(x_i,B) =0},
\]
hence $\tau(x,b)=0$ for all $b\in B$ which means $x\not\ll b$ for all $b\in B$, thus $A{\not\ll_{H}} B$. The case $\displaystyle{\inf_{b\in B}\dist^{+}_{L}(A,b) = 0}$ can be handled in an analogous way.
\end{proof}

\begin{theorem}\label{teo:hiper}
The Lorentzian hyperspace $(\mathcal{H}(X),d_H,\ll_{H},\leq_{H},\tau_{H})$ is a Lorentzian pre-length space.
\end{theorem}

\begin{proof}
By Remark \ref{remarkCont} and Lemma \ref{TauH} (2)  we obtain that $\tau_{H}$ is a continuous function  with respect to $d_H$. It is clear that if $A{\nleq_{H}} B$ then $\tau_{H}(A,B)=0$. Now we focus on the remaining properties of a Lorentzian pre-length space according to Definition \ref{defi:lpls}.

 In fact, by Lemma \ref{tauHmayorcero} we have that if $A\ll_{H} B$ then $\tau_{H}(A,B)>0$ (otherwise $A$ cannot be chronologically related to $B$). On the other hand, if $\tau_{H}(A,B)=r>0$ then for all $a\in A$ and for all $b\in B$ we have $\dist^{-}_{L}(a,B)\geq r$ and $\dist^{+}_{L}(A,b)\geq r$. In conclusion, for all $a\in A$ and $b\in B$ there are two points $a_0\in A$ and $b_0\in B$ ($a_0$ and $b_0$ depending on $a$ and $b$, respectively) with $\tau(a,b_0)\geq r>0$ and $\tau(a_0,b)\geq r>0$, implying $A\ll_{H} B$. Thus $\tau_{H}(A,B)>0$ if and only if $A\ll_{H} B$.

Finally, if $A\leq_{H} B \leq_{H} C$  let $r_1>0$ and $r_2>0$ such that $A\subset \overline{U}^{L,-}_{r_1}(B)$, $B\subset \overline{U}^{L,+}_{r_1}(A)$, and $B\subset \overline{U}^{L,-}_{r_2}(C)$, $C\subset \overline{U}^{L,+}_{r_2}(B)$. Therefore by using point (1) of Lemma \ref{TauH} we have
\[
A \subset \overline{U}^{L,-}_{r_1}(B) \subset \overline{U}^{L,-}_{r_1}\left(\overline{U}^{L,+}_{r_2}(C)\right)\subset \overline{U}^{L,-}_{r_1+r_2}(C),
\]
and similarly $C\subset \overline{U}^{L,+}_{r_1+r_2}(A)$. Thus
\[
r_1+r_2\in\{r>0: A\subset \overline{U}^{L,-}_{r_1+r_2}(C) \mbox{\ and \ } C\subset \overline{U}^{L,+}_{r_1+r_2}(A)\},  
\]
leading to $r_1+r_2\leq \tau_{H}(A,C)$. This last inequality implies 
\[
\tau_{H}(A,B) + \tau_{H}(B,C)\leq \tau_{H}(A,C),
\]
since $r_1$ and $r_2$ were chosen arbitrarily, thus establishing the reverse triangle inequality.

\end{proof}

Let us recall that in a global hyperbolic pre-length space its causal diamonds are compact and its time separation function $\tau$ is finite and continuous. Hence, Proposition \ref{teo:hiper} can be used at once to provide a Lorentzian pre-length structure to the set of causal diamonds of a globally hyperbolic Lorentzian pre-length space. Because of the fundamental role that causal diamonds play in various aspects of mathematical relativity, cuasal theory and Lorentzian geometry, the following is our main result.

\begin{corollary}[The hyperspace of causal diamonds]
Let $(X,d,\ll,\leq,\tau)$ be a globally hyperbolic Lorentzian pre-length space  and define
\[
D(X) = \{J^{+}(x)\cap J^{-}(y): x\leq y\}.
\]
Then $\mathcal{D}(X):=(D(X),d_H,\ll_H,\leq_H,\tau_H)$ is a Lorentzian pre-length space called \textrm{the hyperspace of causal diamonds of} $X$.
\end{corollary}

\begin{example}
 In \cite{StrongGeodesic} the authors show that the hyperspace of closed intervals $(\Sigma(\mathbb{R}),d_H)$ is isometric to the taxicab metric space $(\mathbb{R}\times \mathbb{R}_{\geq 0},d_T)$ where  $d_T((x,t),(y,t))=\vert x-y\vert+\vert t-s\vert $, and the isometry $f:\Sigma(\mathbb{R}) \to \mathbb{R}_{\geq 0}\times \mathbb{R}$ is given by 
\[ 
f([x-t,x+t])=(t,x).
\] 
Further notice that the hyperspaces of causal diamonds $\mathcal{D}(\mathbb{R}^{1}_{1})$ and $(\Sigma(\mathbb{R}),d_H)$ agree as point sets. Thus, the question as if there exists a Lorentzian analog of the isometry $f:\Sigma(\mathbb{R}) \to \mathbb{R}_{\geq 0}\times \mathbb{R}$ depicted above. We answer this question in the affirmative as follows. In order to find a explicit description of $\tau_H([x-t,x+t],[y-s,y+s])$ notice that whenever $[x-t,x+t]\leq_{H} [y-s,y+s]$,  then $x\leq y$, $y-x\geq t-s$ and $y-x\geq s-t$, which in turn implies that
\[
\min\{y-x-(t-s),y-x-(s-t)\}= y-x-\vert t-s\vert .
\]
Now fix $a\in [x-t,x+t]$, then
\[
\begin{array}{rcl}
\dist^{-}_{L}(a,[y-s,y+s]) &=& \sup\{\tau_{\mathbb{R}}(a,c):c\in [y-s,y+s]\} \\ 
&=& \sup\{c-a: c\in [y-s,y+s]\}\\ &=& y+s-a.
\end{array}
\]
Therefore
\[
\displaystyle{\inf_{a\in [x-t,x+t]} \dist^{-}_{L}(a,[y-s,y+s]) = y+s-(x+t)= y-x -(t-s)},
\]
and similarly 
\[
\displaystyle{\inf_{b\in [y-s,y+s]} \dist^{+}_{L}([x-t,x+t],b) = y-s-(x-t) = y-x -(s-t)},
\]
leading to
\[
\begin{array}{rcl}
\tau_H([x-t,x+t],[y-s,y+s]) &=& \min\{y-x-(t-s),y-x-(s-t)\}\\ &=& y-x-\vert t-s\vert . 
\end{array}
\]
In conclusion, the function $f$ is a distance preserving map from the hyperspace of causal diamonds $(\mathcal{D}(\mathbb{R}^{1}_{1}), d_H, \ll_H,\leq_H, \tau_H)$ to the Lorentzian taxicab semi-space $(\mathbb{R}^{2,1}_{T,\geq 0}, d_T,\ll_T, \leq_T,\tau_T)$ described in \cite{solismontes}. 
\end{example}


\section{Applications}\label{sec:apps}

We move on to study in further detail the hyperspace of causal diamonds $\mathcal{D}(\mathbb{R}_1^1\times X)$, where $(X,d)$ is a complete geodesic length space.  As the following results show, in this scenario we have in fact a geodesic and globally hyperbolic Lorentzian length space. We first show that such hyperspace is geodesic and provide explicit parametrizations of maximal geodesic segments connecting causally related diamonds. A few lemmas are in order. First we describe the relations $\ll_H$, $\leq_H$ in explicit terms for Lorentzian taxicab spaces.

\begin{lemma}\label{DiamondsRelated}
Let us consider the diamonds $D_1=J^{+}(t_1,x_1)\cap J^{-}(t_2,x_2)$ and $D_2=J^{+}(a_1,b_1)\cap J^{-}(a_2,b_2)$ for $(t_1,x_1)\leq (t_2,x_2)$ and $(a_1,b_1)\leq (a_2,b_2)$. Then $D_1\leq_H D_2$ if and only if $(t_1,x_1)\leq (a_1,b_1)$ and $(t_2,x_2)\leq (a_2,b_2)$.
\end{lemma}

\begin{proof}
If $D_1\leq_H D_2$, then in virtue of $(t_2,x_2)\in D_1$ there is $(a,b)\in D_2$ with $(t_2,x_2)\leq (a,b) \leq_T (a_2,b_2)$, thus $(t_2,x_2)\leq_T (a_2,b_2)$. Similarly, for $(a_1,b_1)\in D_2$ there is $(t,x)\in D_1$ with $(t_1,x_1)\leq_T (t,x) \leq (a_1,b_1)$, hence $(t_1,x_1)\leq_T (a_1,b_1)$. Now suppose that $(t_1,x_1)\leq_T (a_1,b_1)$ and $(t_2,x_2)\leq_T (a_2,b_2)$ and fix $(t,x)\in D_1$, $(a,b)\in D_2$. Then $(t,x)\leq_T (t_2,x_2)\leq (a_2,b_2)$ and $(t_1,x_1)\leq_T (a_1,b_1)\leq_T (a,b)$, therefore $(t,x)\leq_T (a_2,b_2)$ and $(t_1,x_1)\leq_T (a,b)$. Thus $D_1\leq_H D_2$. 
\end{proof}

\begin{proposition}\label{HausdorffDiamonds}
For every two diamonds $D_1=J^{+}(t_1,x_1)\cap J^{-}(s_1,y_1)$ and $D_2=J^{+}(t_2,x_2)\cap J^{-}(s_2,y_2)$ with $(t_1,x_1)\leq (s_1,y_1)$ and $(t_2,x_2)\leq (s_2,y_2)$, we have
\[
d_H^{T}(D_1,D_2)=\max\{\vert t_2-t_1\vert +d(x_1,x_2),\vert s_2-s_1\vert +d(y_1,y_2)\},
\]
where $d_H^T$ is the Hausdorff distance induced by $d_T$ in $\mathbb{R}_{1}^{1}\times_T X$.
\end{proposition}

The next result is essential in various instances.

\begin{proof}
Let $r\geq 0$ with $D_1\subset\overline{U}_{r}^{T}(D_2) = J^{+}(t_2-r,x_2)\cap J^{-}(s_2+r,y_2)$ and $D_2\subset\overline{U}_{r}^{T}(D_1) = J^{+}(t_1-r,x_1)\cap J^{-}(s_1+r,y_1)$, then $(t_2-r,x_2)\leq (t_1,x_1)\leq (s_1,y_1)\leq (s_2+r,y_2)$ and thus
\[
t_1-(t_2-r)\geq d(x_1,x_2), s_2+r-s_1\geq d(y_1,y_2).
\]
By combining these last two inequalities we obtain
\[
r\geq (t_2-t_1) d(x_1,x_2) \mbox{ \ and \ } r\geq (s_1-s_2) + d(y_1,y_2).
\]
Similarly, when $D_2\subset\overline{U}_{r}^{T}(D_1) = J^{+}(t_1-r,x_1)\cap J^{-}(s_1+r,y_1)$ we obtain
\[
r\geq (t_1-t_2) + d(x_1,x_2) \mbox{ \ and \ } r\geq (s_2-s_1) + d(y_1,y_2).
\]
Therefore
\[
r\geq \vert t_1-t_2\vert + d(x_1,x_2) \mbox{ \ and \ } r\geq \vert s_1-s_2\vert + d(y_1,y_2),
\]
which leads to 
\[
r\geq \max\{\vert t_2-t_1\vert +d(x_1,x_2),\vert s_2-s_1\vert +d(y_1,y_2)\}.
\]
By the definition of $d_H^{T}$ we conclude 
\[
d_H^{T}(D_1,D_2)\geq \max\{\vert t_2-t_1\vert +d(x_1,x_2),\vert s_2-s_1\vert +d(y_1,y_2)\}. 
\]
In order to prove the inequality $d_H^{T}(D_1,D_2)\leq \max\{\vert t_2-t_1\vert +d(x_1,x_2),\vert s_2-s_1\vert +d(y_1,y_2)\}$ just note that 
$\delta=\max\{\vert t_2-t_1\vert +d(x_1,x_2),\vert s_2-s_1\vert +d(y_1,y_2)\}$ satisfies $D_1\subset \overline{U}_{\delta}^{T}(D_2)$ and $D_2\subset \overline{U}_{\delta}^{T}(D_1)$. 
\end{proof}

We move on to prove that $\mathcal{D}(\mathbb{R}_1^1\times X)$ is a geodesic Lorentzian pre-Length space and provide explicit $\tau$-parametrizations for maximal curves. We follow the same approach as in Lemma \ref{prop:Hgeo}. Thus we need to show first that the $d_T$ tubular neighborhood and that the intersection of two causal diamonds  are causal diamonds themselves. The following lemmas deal with these issues.

\begin{lemma}
Let $(t,x)\leq_T (s,y)$ be two points in $\mathbb{R}_{1}^{1}\times_T X$ and $D=J^{+}(t,x)\cap J^{-}(s,y)$. Then for every $r\geq 0$ we have
\[
\overline{U}_{r}^{T}(D) = J^{+}(t-r,x)\cap J^{-}(s+r,y),
\]
where $\overline{U}_{r}^{T}=\{(u,z): {\dist}_T((u,z),D)\leq r\}$ is the closed $r-$tubular neighborhood with respect to $d_T$.
\end{lemma}

\begin{proof}
First of all, if $(u,z)\in \overline{U}_{r}^{T}(D)$ then due to the compactness of $D$ there exists $(u_0,z_0)\in D$ with 
\[
u_0-u + d(z,z_0) \leq \vert u-u_0\vert + d(z,z_0) \leq d_T((u,z),(u_0,z_0))\leq r. 
\]
We also have $s-u_0\geq d(y,z_0)$ and $u_0-t\geq d(x,z_0)$ since $(t,x)\leq_T (u_0,z_0)\leq_T (s,y)$. Thus
\[
d(x,z)\leq d(x,z_0)+ d(z_0,z)\leq (u_0-t) + (r+u-u_0) = u-(t-r),
\]
implying $(t,x)\leq_T (u,z)$. Similarly $(u,z)\leq_T (s,y)$ and in conclusion $(u,z)\in J^{+}(t-r,x)\cap J^{-}(s+r,y)$. Conversely, suppose $(u,z)\in J^{+}(t-r,x)\cap J^{-}(s+r,y)$. If $(u,z)\in D$ then $(u,z)\in \overline{U}_{r}^{T}(D)$, otherwise $(t,x){\nleq}_T (u,z)$ or $(u,z){\nleq}_T(s,y)$.  Assume $(u,z){\nleq}_T(s,y)$ then $s-u<d(y,z)$. If $0<s-u$, let $\gamma:[0,d(y,z)]\to X$ be a geodesic segment connecting $y$ with $z$ (that is, $\gamma(0)=y$ and $\gamma(d(y,z))=z$) and take $z_0=\gamma(s-u)$, thus 
\[
d(y,z_0)=d(\gamma(0),\gamma(s-u))= s-u,
\]
therefore $(u,z_0)\leq_T (s,y)$ and $(u,z_0)\in D$. Moreover,
\[
d_T((u,z),(u,z_0)) = d(z,z_0) = d(y,z)-(s-u)\leq r,
\]
since $(u,z)\leq_T (s+r,y)$. Hence $(u,z)\in \overline{U}_{r}^{T}(D)$. Alternatively, if $s-u<0$ then 
\[
d_T((u,z),(s,y)) = u-s + d(z,y) \leq (u-s) + (s+r-u) = r,
\]
again by $(u,z)\leq_T (s+r,y)$. So in this case we have $(u,z)\in \overline{U}_{r}^{T}(D)$ as well. A similar argument proves that if  $(t,x){\nleq}_T(u,z)$ then $(u,z)\in \overline{U}_{r}^{T}(D)$, thus finishing the proof. 
\end{proof}

\begin{lemma}\label{CausalCurves}
Let $D_1,D_2\in \mathcal{D}(\mathbb{R}^{1}_{1}\times X)$, $r=d_H(D_1,D_2)>0$ and $D_1\leq_H D_2$. Let us denote $D_1=J^{+}(t_1,x_1)\cap J^{-}(s_1,y_1)$ and $D_2=J^{+}(t_2,x_2)\cap J^{-}(s_2,y_2)$. If $u\in [0,r]$ then 
\[
\overline{U}_{u}(D_1)\cap \overline{U}_{r-u}(D_2) = J^{+}(t_2+u-r,x_2)\cap J^{-}(s_1+u,y_1).
\]
\end{lemma}

\begin{proof}
In virtue of Lemma \ref{DiamondsRelated} it is enough to prove that $(t_2+u-r,x_2)\leq_T (s_1+u,y_1)$, which is equivalent to $s_1-t_2+r\geq d(x_2,y_1)$. First, applying Proposition \ref{HausdorffDiamonds} we way suppose that
\[
r=d_H(D_1,D_2)=t_2-t_1-d(x_1,x_2).
\]
But $s_1-t_1\geq d(x_1,y_1)$, thus
\[
s_1-t_2+r = s_1-t_1+d(x_1,x_2) \geq d(y_1,x_1) + d(x_1,x_2)\geq d(y_1,x_2).
\]
\end{proof}

\begin{proposition}\label{PathConnected}
Let $D_1,D_1\in \mathcal{D}(\mathbb{R}^{1}_{1}\times X)$, $r=d_H(D_1,D_2)>0$ and $D_1\leq_H D_2$. Then $\gamma:[0,r]\to \mathcal{D}(\mathbb{R}^{1}_{1}\times X)$ defined as
\[
\gamma(u)=\overline{U}_{u}(D_1)\cap \overline{U}_{r-u}(D_2)
\]
is a causal curve connecting $D_1$ with $D_2$ and $\tau_H(\gamma(u),\gamma(v))=\vert u-v\vert $ for every $u,v\in [0,r]$.
\end{proposition}

\begin{proof}
The first part of the proof is an immediate consequence of Lemma \ref{CausalCurves}. Now we focus on the second part. Let $0\leq u \leq v \leq r$, then by using Lemma \ref{TauH}
\[
\tau_H(\gamma(u),\gamma(v))=\min\{v-u,v-u\}= v-u.
\]
In particular observe that $\tau_H(D_1,D_2)=r$.
\end{proof}

We now establish the main result.

\begin{theorem}
Let $(X,d)$ be a complete geodesic length space. The hyperspace of causal diamonds $\mathcal{D}(\mathbb{R}^{1}_{1}\times_T X)$ is a geodesic globally hyperbolic Lorentzian length space.
\end{theorem}

\begin{proof}
Local causal closedness in $\mathcal{D}(\mathbb{R}^{1}_{1}\times_T X)$ follows easily from the corresponding property of $\mathbb{R}^{1}_{1}\times_T X$ and Lemma \ref{DiamondsRelated}. Moreover,   Proposition \ref{PathConnected} guarantees that $\mathcal{D}(\mathbb{R}^{1}_{1}\times_T X)$ is causally path connected and intrinsic. Thus, to prove that the hyperspace of causal diamonds is a Lorentzian length space  it only remains to prove that $\mathcal{D}(\mathbb{R}^{1}_{1}\times_T X)$ is localizable. Indeed, we will show that 
\[
\Omega^{r}_{D} = I^{+}_{H}(D^{-}_{r})\cap I^{-}_{H}(D^{+}_{r})
\]
is a localizing neighborhood, where $D=J^{+}(t,x)\cap J^{-}(s,y)$, $r>0$, $D^{-}_{r}=J^{+}(t-r,x)\cap J^{-}(s-r,y)$ and $D^{+}_{r}=J^{+}(t+r,x)\cap J^{-}(s+r,y)$. First, let $\gamma:[a,b]\to \mathcal{D}(\mathbb{R}^{1}_{1}\times_T X)$ to be a causal curve contained in $\Omega^{r}_{D}$, where 
\[
\gamma(u)=J^{+}(\alpha(u),\beta(u))\cap J^{-}(\tilde{\alpha}(u),\tilde{\beta}(u))
\]
and $\alpha, \tilde{\alpha}:[a,b]\to \mathbb{R}$, $\beta, \tilde{\beta}:[a,b]\to X$. Now, for a partition $a=u_0<u_1<\cdots < u_n=b$ observe that $\alpha(u_{i+1})-\alpha(u_{i})\geq   d(\beta(u_{i+1}),\beta(u_{i}))$ and $\tilde{\alpha}(u_{i+1})-\tilde{\alpha}(u_{i}) \geq d(\tilde{\beta}(u_{i+1}),\tilde{\beta}(u_{i}))$ by Lemma \ref{DiamondsRelated}. Then, applying Lemma \ref{HausdorffDiamonds} we obtain that
\[
\begin{array}{rcl}
\displaystyle{\sum_{i=0}^{n-1} d^{T}_{H}(\gamma(u_i),\gamma(u_{i+1}))} &=& \displaystyle{\sum_{i=0}^{n-1} \max\{\vert\alpha(u_{i+1})-\alpha(u_{i})\vert + d(\beta(u_{i+1}),\beta(u_{i}))}, \\
&& \quad\quad\quad\quad \vert\tilde{\alpha}(u_{i+1})-\tilde{\alpha}(u_{i})\vert + d(\tilde{\beta}(u_{i+1}),\tilde{\beta}(u_{i}))\} \\
&\leq& \displaystyle{\sum_{i=0}^{n-1} 2(\alpha(u_{i+1})-\alpha(u_{i})) + 2(\tilde{\alpha}(u_{i+1})-\tilde{\alpha}(u_{i}))} \\
&=& 2\left( \alpha(b)-\alpha(a) + \tilde{\alpha}(b)-\tilde{\alpha}(a) \right).
\end{array}
\]
On the other hand, since $D^{+}_{r}\leq_{H} \gamma(a)\leq_{H} \gamma(b) \leq_{H} D^{+}_{r}$, it follows that 
\[
\alpha(b)-\alpha(a) + \tilde{\alpha}(b)-\tilde{\alpha}(a) \leq 4r.
\]
Therefore $L_{d^{T}_{H}}(\gamma)\leq 8r$ for every causal curve contained in $\Omega^{r}_{D}$. It is clear that $(\Omega^{r}_{D}, d^{T}_{H},\ll_{H},\leq_{H}, \omega^{r}_{D},\tau_H\vert_{\Omega^{r}_{D}})$ is a Lorentzian pre-length space. If $D_0=J^{+}(t_0,x_0)\cap J^{-}(s_0,y_0)\in \Omega^{r}_{D}$ we need to see that $I^{+}_{H}(D_0)\cap \Omega^{r}_{D}\neq \varnothing$. In order to do this, observe that if $m_1,m_2 \in X$ are some middle points for $x, x_0$ and $y, y_0$, respectively, then $D'=J^{+}\left(\frac{t+r+t_0}{2},m_1\right)\cap J^{-}\left(\frac{s+r+s_0}{2},m_2\right)$ belongs to $I^{+}_{H}(D_0)\cap \Omega^{r}_{D}$. Similarly we can prove that $I^{-}_{H}(D_0)\cap \Omega^{r}_{D}\neq \varnothing$ for every $D_0\in \Omega^{r}_{D}$. 

First, we need to prove that every causal diamond $J^{+}_{H}(D_1)\cap J^{-}_{H}(D_2)$ is compact with respect to the topology induced by $d_H$ on $\mathcal{D}(\mathbb{R}^{1}_{1}\times X)$, where $D_1=J^{+}(t_1,x_1)\cap J^{-}(s_1,y_1)$ and $D_2=J^{+}(t_2,x_2)\cap J^{-}(s_2,y_2)$. Let $\{F_n=J^{+}(p_n,z_n)\cap J^{-}(q_n,w_n)\}$ a sequence of diamonds contained in $J^{+}_{H}(D_1)\cap J^{-}_{H}(D_2)$ with $F_n\neq \varnothing$ for every $n\geq 1$. Then $D_1\leq_H F_n \leq D_2$ and by applying Lemma \ref{DiamondsRelated} we obtain $(t_1,x_1)\leq_T (p_n,z_n)\leq (t_2,x_2)$ and $(s_1,y_1)\leq_T (q_n,w_n)\leq (s_2,y_2)$. It follows that $(p_n,z_n)\in J^{+}(t_1,x_1)\cap J^{-}(t_2,x_2)$ and $(q_n,w_n)\in J^{+}(s_1,y_1)\cap J^{-}(s_2,y_2)$. By using the compactness of $J^{+}(t_1,x_1)\cap J^{-}(t_2,x_2)$ and $J^{+}(s_1,y_1)\cap J^{-}(s_2,y_2)$ and a diagonal Cantor's process we are able to find two convergent subsequences $\{(p_{n_k},z_{n_k})\}_{k\in\mathbb{N}}$ and $\{(q_{n_k},w_{n_k})\}_{k\in\mathbb{N}}$ converging to $(p,z)$ and $(q,w)$, respectively. Moreover, the sequence $\{F_{n_k}\}_{k\in\mathbb{N}}$ converges to $F=J^{+}(p,z)\cap J^{-}(q,w)$. Since $\leq_T$ is closed for every diamond in $\mathbb{R}^{1}_{1}\times_T X$ we have that $(p,z)\in J^{+}(t_1,x_1)\cap J^{-}(t_2,x_2)$ and $(q,w)\in J^{+}(s_1,y_1)\cap J^{-}(s_2,y_2)$. This last statement implies that $F\in J_H^{+}(D_1)\cap J_H^{-}(D_2)$.
\end{proof}


We end this section by providing a nice realization of $\mathcal{D}(\mathbb{R}^{1}_{1}\times_T X)$, that in passing reveals a close connection between the Lorentzian taxicab and uniform products.

\begin{theorem}
Let $f: \mathcal{D}(\mathbb{R}^{1}_{1}\times_{T} X) \to X\times_{\infty} X$ defined as $f(D)=((t,x),(s,y))$, where $D=J^{+}(t,x)\cap J^{-}(s,y)$. Then $f$ is an isometry in $f(D)$ and a $\tau$-preserving map.
\end{theorem}
\begin{proof}
That $f$ is a metric isometry follows directly from its definition and Proposition \ref{HausdorffDiamonds}. Moreover, $f$ is $\tau$ preserving as the following result shows.
\end{proof}

\begin{proposition}\label{HausdorffLorentzian}
Let $D_1=J^{+}(t_1,x_1)\cap J^{-}(s_1,y_1)$ and $D_2=J^{+}(t_2,x_2)\cap J^{-}(s_2,y_2)$ two diamonds with $(t_1,x_1)\leq_T (s_1,y_1)$, $(t_2,x_2)\leq_T (s_2,y_2)$ and $D_1\leq_{H} D_2$, then
\[
\tau_H(D_1,D_2) = \min\{ t_2-t_1-d(x_1,x_2) , s_2-s_1-d(y_1,y_2) \}.
\]
\end{proposition}

\begin{proof}
Let $(u_1,z_1)\in D_1$, then 
\[
\begin{array}{rcl}
\dist_{L}^{-}((u_1,z_1),D_2) &=& \sup\{\tau((u_1,z_1),(u_2,z_2)): (u_2,z_2)\in D_2\} \\
&=& \sup\{\tau((u_1,z_1),(u_2,z_2)): (u_1,z_1) \leq (u_2,z_2)\in D_2\}.
\end{array}
\]
Thus, for $ (u_1,z_1) \leq (u_2,z_2)\leq (s_2,y_2)$ we obtain
\[
\begin{array}{rcl}
\tau((u_1,z_1),(u_2,z_2)) &\leq& \tau((u_1,z_1),(u_2,z_2))+ \tau((u_2,z_2), (s_2,y_2)) \\
&\leq& \tau((u_1,z_1),(s_2,y_2)).
\end{array}
\]
In conclusion we must have 
\[
\dist_{L}^{-}((u_1,z_1),D_2)=  \tau((u_1,z_1),(s_2,y_2)).
\]
Since $(u_1,z_1)\leq (s_1,y_1)\leq (s_2,y_2)$, we can show similarly that 
\[
\displaystyle{\inf_{(u_1,z_1)\in D_1}\dist_{L}^{-}((u_1,z_1),D_2) = \tau((s_1,y_1),(s_2,y_2)) = s_2-s_1-d(y_1,y_2).}
\]
Also, a similar process lead us
\[
\displaystyle{\inf_{(u_2,z_2)\in D_2}\dist_{L}^{+}(D_1,(u_2,z_2)) = \tau((t_1,x_1),(t_2,x_2)) = t_2-t_1-d(x_1,x_2).}
\]
Then, by the way we defined $\tau_H$ we have
\[
\tau_H(D_1,D_2) = \min\{t_2-t_1-d(x_1,x_2), s_2-s_1-d(y_1,y_2)\}.
\] 
\end{proof}

\section*{Acknowledgments}

W. Barrera acknowledges the support of Conacyt under grants SNI 45382 and Ciencia de Frontera 21100. L. Montes recognizes the support of Conacyt under the Becas Nacionales program (783177). D. Solis was partially supported by Conacyt SNI 38368 and UADY-FMAT PTA 2023.

\begin{enumerate}
\item \textsc{Waldemar Barrera. Facultad de Matem\'aticas. Universidad Aut\'onoma de Yuct\'an, Perif\'erico Norte 13615. M\'erida, M\'exico. } bvargas@correo.uady.mx

\item \textsc{Luis Montes de Oca. Facultad de Matem\'aticas. Universidad Aut\'onoma de Yuct\'an, Perif\'erico Norte 13615. M\'erida, M\'exico. } mauricio.montes@alumnos.uady.mx

\item \textsc{Didier A. Solis. Facultad de Matem\'aticas. Universidad Aut\'onoma de Yuct\'an, Perif\'erico Norte 13615. M\'erida, M\'exico. }  didier.solis@correo.uady.mx

\end{enumerate}

\end{document}